\documentclass[journal]{IEEEtran} 

\usepackage[dvips]{color}
\usepackage{epsf}
\usepackage{times}
\usepackage{epsfig}
\usepackage{graphicx}
\usepackage{float} 
\usepackage{amsmath}
\usepackage{amssymb}
\usepackage{amsxtra}
\usepackage{amsthm}
\usepackage{bbm}

\usepackage{here}
\usepackage{rawfonts}
\usepackage{times}
\usepackage{url}
\usepackage{cite}
\usepackage{comment}
\usepackage[utf8]{inputenc}
\usepackage{caption}
\usepackage{subcaption}
\usepackage[normalem]{ulem}

\usepackage{pstricks}
\usepackage{algorithm}
\usepackage{algpseudocode}
\usepackage{lipsum}

\newtheorem{theorem}{\bf Theorem}

\IEEEoverridecommandlockouts 

\begin{document} 
\title{\huge
    Privacy Preserving Semantic Communications in Wireless Edge Networks with Vision Language Models
} 
\author{ 

    Haoran Chang, \IEEEmembership{Graduate Student Member, IEEE},
    Mingzhe Chen, \IEEEmembership{Senior Member, IEEE},
    and Qianqian Zhang, \IEEEmembership{Member, IEEE}

}

\IEEEaftertitletext{\vspace{-1\baselineskip}}
\maketitle
\begingroup
\renewcommand{\thefootnote}{}

\footnotetext{Haoran Chang and Qianqian Zhang are with the Department of
Electrical and Computer Engineering, Rowan University, Glassboro, NJ 08028, USA
(E-mail: \url{changh35@rowan.edu}; \url{zhangqia@rowan.edu}).}

\footnotetext{Mingzhe Chen is with the Department of Electrical and Computer
Engineering, University of Miami, Coral Gables, FL 33146, USA
(E-mail: \url{mingzhe.chen@miami.edu}).  He is also with the Frost Institute for Data Science and Computing, University of Miami,  FL, USA.}

\endgroup

\begin{abstract} 
Semantic communication has emerged as a promising paradigm for next-generation wireless systems by improving communication efficiency through the transmission of high-level semantic features rather than raw bits. 
Recent advances in  semantic communications have enabled  multiple  devices to collaboratively transmit complementary semantic information, thereby enhancing the semantic understanding and reasoning at the receiver. 
Meanwhile, the increasing availability of multi-modal data has further enriched the semantic representations and expanded the scope of wireless applications. 
However, the involvement of  collaborative devices and multimodal transmissions substantially enlarges the attack surface and introduces new privacy risks, as sensitive information may be exposed through both inter-device semantic fusion and cross-modal semantic leakage during transmission.
To address this challenge, we propose a privacy-preserving  semantic communication framework to protect privacy-sensitive information in the wireless edge communication system. 
Leveraging a vision-language model (VLM), the proposed framework extracts textual semantics from the input image  and identifies privacy-sensitive entities using a privacy database maintained exclusively at an edge server. 
Prior to image transmission, each edge device removes the identified private regions while preserving  useful semantic content to support unknown downstream tasks. 
Upon reception, the server reconstructs the removed regions from the received masked image, guided by textual embeddings and semantic priors learned by the VLM.
The framework is designed against a strong model-aware adversary that can eavesdrop on  wireless transmissions  and has full knowledge of the edge-device model parameters, but has no access to any server-side data. 
Simulation results show that the proposed framework reduces privacy leakage to the model-aware adversary by more than $50\%$ compared with the state-of-art semantic communication scheme without privacy protection. 
Moreover, the authorized server achieves a $48\%$ improvement in perceptual reconstruction quality over the adversary when recovering the removed private regions.
The proposed framework also effectively suppresses cross-device semantic redundancy, with the estimated mutual information between transmitted representations approaching $0$ bit, indicating that redundant semantic information is rarely transmitted across devices. 

\end{abstract}  
 

\section{Introduction}

In recent years, the rapid expansion of wireless edge network applications and the growing volume of data have placed increasing demands on efficient transmission under limited communication resources. In this context, semantic communication has attracted significant attention. Unlike traditional communication systems that focus on the reliable delivery of bit sequences, semantic communication aims to convey the underlying meaning of the transmitted information \cite{xie2021deep}. Based on the received semantic information, the receiver is able to perform the downstream tasks or recover the original content. As a result, semantic communication offers a promising solution for wireless edge communication systems. 

To realize the promise of semantic communications, deep neural networks (DNNs) have become a key enabling technology.
A DNN-based framework  was first  proposed  in \cite{xie2021deep} to jointly train semantic and channel encoders/decoders for semantic extraction and robust text transmission over physical channels. 
For image transmissions,   variational autoencoder (VAE) \cite{hu2023robust} and vision transformer \cite{yoo2022real}  have been employed for image feature extraction and  reconstruction. 
However, most early works primarily relied on single-modal representations with limited semantic reasoning and cross-modal  understanding. 
To address these challenges, some recent works  \cite{cicchetti2024language} and \cite{zhao2024lamosc} leveraged  the 
large-language model (LLM) 
to enhance the image  understanding and reconstruction. 
In these frameworks, the transmitter  converts input images into textual descriptions, 
while the receiver reconstructs the image   from texts and transmitted semantic representations  to achieve high reconstruction quality. 
Nevertheless, existing studies  \cite{li2017person,zhao2021weakly,jiang2024survey,agyeman2024decentralized} have shown that textual descriptions alone can reveal sufficient identity-related semantic information to enable image retrieval, leading to privacy leakage. 
These findings highlight that, despite the remarkable semantic capabilities of multimodal frameworks, protecting privacy-sensitive information remains a fundamental challenge for semantic communications. 


Beyond  multi-modal presentations, recent works in \cite{xie2022task} and \cite{shao2023task} have extended semantic communications  from  point-to-point transmission into  multi-device systems, where multiple devices collaboratively transmit semantic information about the same event. 
By exploiting diverse sensing perspectives, these frameworks overcome  the  sensing limitation of  individual devices and improve the performance of downstream tasks. 
However, the semantic information generated by different devices is often partially redundant, resulting in   unnecessary communication overhead. 
More critically, inter-device semantic correlation substantially enlarges the attack surface, allowing adversaries to infer sensitive information by jointly analyzing transmissions from multiple devices or launching model inversion attacks on compromised edge devices with limited security capabilities \cite{yao2026privguardinfer,khowaja2026post,dong2022privacy}. 
Consequently, secure semantic communications in the wireless edge network require not only efficient suppression of redundant information but also effective protection against semantic leakage. 


\subsection{Related Work} 

\subsubsection{Multimodal-enabled Semantic Communication}

Vision-language model (VLM) have recently become a key technology for multi-modal semantic communication.  
By representing visual content as natural language, VLMs provide semantically interpretable representations that improve communication efficiency and robustness to channel impairments  \cite{VLM-first}. 
At the receiver, the transmitted textual semantics are used to condition a generative model for image reconstruction
However, textual descriptions alone often fail to preserve fine-grained visual details, resulting in limited reconstruction quality. 
To address this limitation, \cite{cicchetti2024language} and \cite{zhao2024lamosc} additionally transmit latent features to enhance reconstruction fidelity. 
While these works demonstrate the effectiveness of VLMs for multimodal semantic communications, they primarily focus on communication efficiency and reconstruction quality, overlooking the privacy risks associated with semantic representations.


\subsubsection{Privacy-preserving Semantic Communication}

Existing  privacy-preserving semantic communication approaches can be categorized into differential privacy, cryptographic protection, and sensitive information suppression \cite{meng2025survey}. 
Differential privacy  protects semantic representations by introducing calibrated perturbations. 
For example, \cite{chen2024enhancing} separates  semantic features into private and non-private latent representations and perturbs only the private component before transmission, 
while \cite{seif2024collaborative} employs differential privacy to protect features during collaborative inference across multiple devices. 
Although these methods achieve a favorable tradeoff between privacy preservation and task performance, they primarily consider conventional wireless eavesdroppers and overlook potential learning-based attacks. 

Cryptographic approaches, on the other hand, secure semantic representations through learned encryption mechanisms. 
Neural adversarial cryptography was first introduced in \cite{abadi2016learning} and later adopted for semantic communications in \cite{luo2023encrypted}, where a pre-shared secret key enables legitimate receivers to recover the transmitted information while preventing unauthorized decoding. 
However, these methods rely on secure key establishment and management, which remain challenging in practice. 
Finally, sensitive information suppression removes or obfuscates privacy-sensitive content before transmission.  Representative examples include suppressing sensitive semantic features through adversarial learning \cite{wang2023privacy} and removing user-specified objects from shared images \cite{bai2023precise}.
However, these methods are largely task-specific, with semantic encoder and decoder optimized to predefined tasks, and fail to support reconstruction of the original data. 
Thus, this limitation restricts the utility of the received data in task-unaware scenarios.

\subsubsection{Multi-Device Information Bottleneck} 

Multiple devices often observe the same scene or object from different perspectives, resulting in redundant semantic information.  To reduce address this issue, the information bottleneck principle has  been extended to multi-device settings \cite{aguerri2019distributed}. 
This work established a theoretical framework for balancing communication efficiency and task information. 
Building upon this framework, \cite{shao2023task} leveraged the  information bottleneck for multi-device task-oriented communication, where each device encodes its local observations and transmits compact semantic representations to a central server for inference. 
Similarly, \cite{wei2023federated} incorporated the information bottleneck  into  federated learning, where the devices are trained jointly with a server. 
Although these works demonstrate the effectiveness of information bottleneck learning for distributed inference, they primarily aim to preserve task-relevant information and do not explicitly address semantic redundancy across multiple devices. 
Nevertheless, existing studies provide valuable insights for designing efficient multi-device semantic  systems.

\subsection{Main Contributions}

In this work, a privacy-preserving semantic communication framework is proposed  to jointly protect sensitive information and reduces cross-device  semantic redundancy in  wireless edge networks. Main contributions  are summarized as follows:

\begin{itemize}
    \item A VLM-enabled  semantic communication framework is proposed for wireless edge systems. 
    A text-based privacy identification module  detects privacy-sensitive entities by matching image descriptions against a server-side privacy database. 
    The identified image regions are removed before transmission and reconstructed only at the authorized receiver using a VLM, preventing unauthorized receivers from recovering the protected content.
    \item To prevent information leakage, we propose an encrypted semantic text transceiver. 
    Instead of relying on pre-shared secret keys, the encryption key is generated from reciprocal wireless channel characteristics between the legitimate edge device and server, eliminating the need for prior key distribution. 
    As a result, unauthorized receivers cannot correctly decode intercepted semantic messages, even with full knowledge of the edge-device model parameters.


    \item We propose the information bottleneck framework to suppress redundant semantic information across multiple edge devices. 
    A variational approximation  is developed to obtain a tractable optimization problem of minimizing the cross-device mutual information  while preserving task-relevant semantic content. 

    \item Extensive simulation results demonstrate that the proposed framework achieves high-quality image reconstruction at the authorized receiver while effectively preventing model-aware adversaries from recovering privacy-sensitive content. Furthermore, the proposed distributed semantic information bottleneck significantly reduces cross-device semantic redundancy, thereby improving semantic communication efficiency.
    



\end{itemize}

The remainder of this paper is organized as follows. 
Section II presents the  system model, including the privacy detection and protection procedures, the model-aware adversary, the  semantic information bottleneck, and the problem formulation. 
Section III details the proposed solution. 
Section IV presents the simulation evaluation of the proposed framework. 
Finally, Section V concludes the paper.

\section{System Model and Problem Formulation}\label{sysModel_proFormulation}

\begin{figure*}[t]
    \centering
    \includegraphics[width=0.85\textwidth]{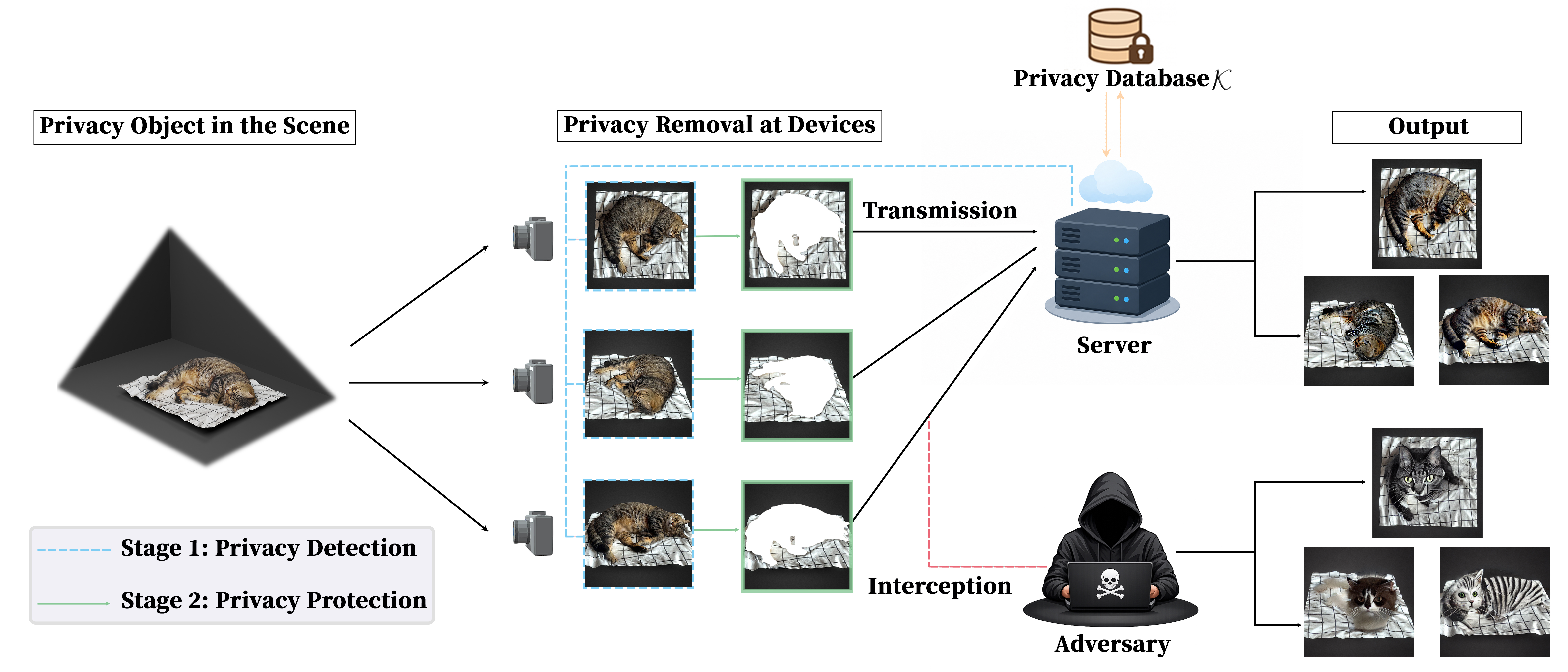} \vspace{-0.2cm}
    \caption{An example of the proposed privacy-preserving framework. The devices remove sensitive content before transmission. The server reconstructs the privacy object, while the adversary cannot reliably recover its identity.}
    \label{fig:example}
\end{figure*}

We consider a wireless system comprising a group of edge devices $\mathcal{N}$ and a server. Each device $n \in \mathcal{N}$ captures an image of an object of interest 
from a different viewpoint and transmits the data to the server. 
In this context, the object of interest often constitutes sensitive information, which unauthorized entities may attempt to intercept and extract. 
To safeguard privacy, a semantic protection framework is employed.  
Specifically, a privacy dataset $\mathcal{K}$ is maintained exclusively at the server to characterize and identify sensitive content within the captured images. 
The edge devices, however, have no prior knowledge of the sensitive object and cannot directly access $\mathcal{K}$. Instead, each device relies on semantic guidance provided by the server using text-based instructions. 
Based on this guide, each device performs on-device filtering to remove sensitive information from captured images prior to transmission, ensuring that no privacy-related content is included in the transmitted data. 
For example, if the privacy object is a cat named Vicky, as shown in Fig.~\ref{fig:example}. After privacy removal, only the server can semantically recover the cat’s identity, while the adversary produces inconsistent outputs. Although the adversary may infer that the masked image is a cat, it cannot clearly identify the color or fur pattern, thus preserving the identity information about Vicky. 
To implement this mechanism, we first introduce the privacy detection module, privacy protection framework, and adversary model, followed by the performance metrics and  problem formulation. 

\begin{figure*}[t]
    \centering
    \makebox[\textwidth][l]{
        \hspace*{-3.20cm}
        \includegraphics[width=1.19\textwidth]{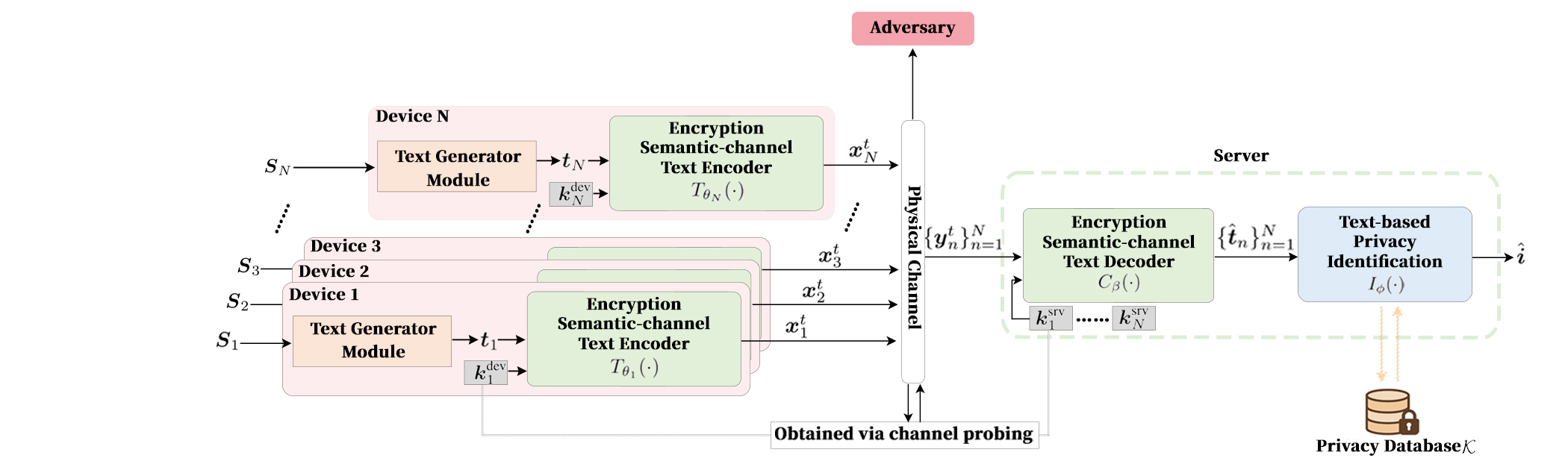}
    }
    \vspace{-0.5cm}
    \caption{Privacy detection module to identify sensitive content via text descriptions.}
    \label{fig:system_model_text}
\end{figure*}

\subsection{Privacy Detection}  
The purpose of privacy detection is to identify sensitive information in the captured images of each edge device with the assistance of the server. 
Let 
$\mathcal{S}=\{\boldsymbol{S}_n\}_{n=1}^{N}$ 
denote a multi-view image set captured by  devices, where $\ \boldsymbol{S}_n \in$ $\mathbb{R}^{C \times H \times W}$ denotes the image at device $n \in \mathcal{N}$, and $\mathit{C}$, $\mathit{H}$, and $\mathit{W}$ represent the channel, height, and width, respectively. 
As shown in Fig. \ref{fig:system_model_text}, each device  employs a Bootstrapping Language-Image Pre-training (BLIP) model \cite{li2022blip } to generate a textual description of the captured image,  as
\begin{equation}
\boldsymbol{t}_n = f_{\mathrm{BLIP}}(\boldsymbol{S}_n)= \left[ w_{n,1},\, w_{n,2},\, \ldots,\, w_{n,L} \right] ,
\label{equ:s2t}
\end{equation}
where $w_{n,l}$ represents the $l$-th word generated by device $n$. The resulting token sequence $\boldsymbol{t}_n $ serves as a compact semantic representation of the image, which is then transmitted to the server for privacy detection via a frequency division multiple access approach. 

To  protect textual information, a physical-layer encryption module is employed, where the channel state information (CSI) between each device and the server is first estimated using pilot signaling and then used for text-massage protection.  
Assuming a full-duplex  setting, the CSI observed at  device $n$ and the server are denoted by $\mathbf{h}_n^\text{dev}$ and $\mathbf{h}_n^\text{srv}$, respectively. 
Due to the channel reciprocity, the generated  secret keys $\boldsymbol{k}_n^\text{dev} = f_{\mathrm{PLK}}(\mathbf{h}_n^\text{dev}) \in \mathbb{R}^d$ at device $n$ and $\boldsymbol{k}_n^\text{srv} = f_{\mathrm{PLK}}(\mathbf{h}_n^\text{srv})\in \mathbb{R}^d$ at the server are highly correlated \cite{mathur2008radiotelepathy}, 
where  $d$ is the embedding dimension of the textual token representation and $f_{\mathrm{PLK}}(\cdot)$ is the physical-layer key generation function \cite{zeng-plkg-2015}.

Each device then encodes the text data using an encryption encoder $T_{\theta_n}(\cdot)$ with $\theta_n$ as the trainable parameters of device $n$, and $\boldsymbol{\Theta} = [\theta_1,   \ldots, \theta_N]$ represents the  parameters for all devices. The transmitted  signal from  device $n$ is given by
\begin{equation}
    \boldsymbol{x}_n^{t} = T_{\theta_n}(\boldsymbol{t}_n, \boldsymbol{k}_n^\text{dev}). 
    \label{equ:t2x}
\end{equation}
The received signal from device $n$ at the server is 
\begin{equation}
    \boldsymbol{y}_{n}^{t}
    =
    \boldsymbol{h}_{n}
    \odot
    \boldsymbol{x}_{n}^{t}
    +
    \boldsymbol{n},
    \label{equ:x2yt}
\end{equation}
where $\boldsymbol{h}_n$ denotes the flat-fading channel  from device $n$ to the server, $\odot$ denotes element-wise multiplication, and $\boldsymbol{n} \sim \mathcal{CN}(0, \sigma_n^{2} \mathbf{I})$ is additive white Gaussian noise.
The server then applies a decoder $C_{\beta}(\cdot)$  with parameters $\beta$ to recover the text: 
\begin{equation}
\hat{\boldsymbol{{t}}}_n = C_{\beta}({\boldsymbol{y}}_{n}^{t}, \boldsymbol{k}_n^\text{srv}).
\label{equ:yk2t}
\end{equation}

After aggregating the reconstructed semantic texts $\{ \boldsymbol{\hat{t}}_1, \cdots, \boldsymbol{\hat{t}}_N\}$ from all  devices,  the server  applies a text-based privacy identification module $I_{\phi}(\cdot)$, parameterized by $\phi$, to determine whether the  captured image at each device contains any sensitive content defined in the privacy database $\mathcal{K}$.  
Here, the privacy database $\mathcal{K} = \left\{ \mathcal{S}^{*}_i \right\}_{i=1,\cdots,K}$  is constructed from predefined privacy entities specified by the server,  
where
$\mathcal{S}^{*}_i = \{ \boldsymbol{S}^{*}_{i,m} \}_{m=1,\cdots,M}$ is a set of $M$ sample images associated with each privacy identity $i$. 
The privacy dataset is typically established during system initialization, device registration, or deployment via highly secure links. 

At the end of the privacy detection stage, if no sensitive object is identified, the corresponding image is transmitted   to the server without further processing.  
Otherwise, a  privacy identity label $\hat{i} = I_{\phi}(\{\hat{\boldsymbol{t}}_n\}_{n=1}^{N})$ is predicted based on the multi-view textual descriptions, and 
a feedback signal is  broadcast to all devices to indicate whether privacy protection  is required. 

\begin{figure*}[t]
    \centering
    \makebox[\textwidth][l]{
        \hspace*{-3.55cm}
        \includegraphics[width=1.2\textwidth]{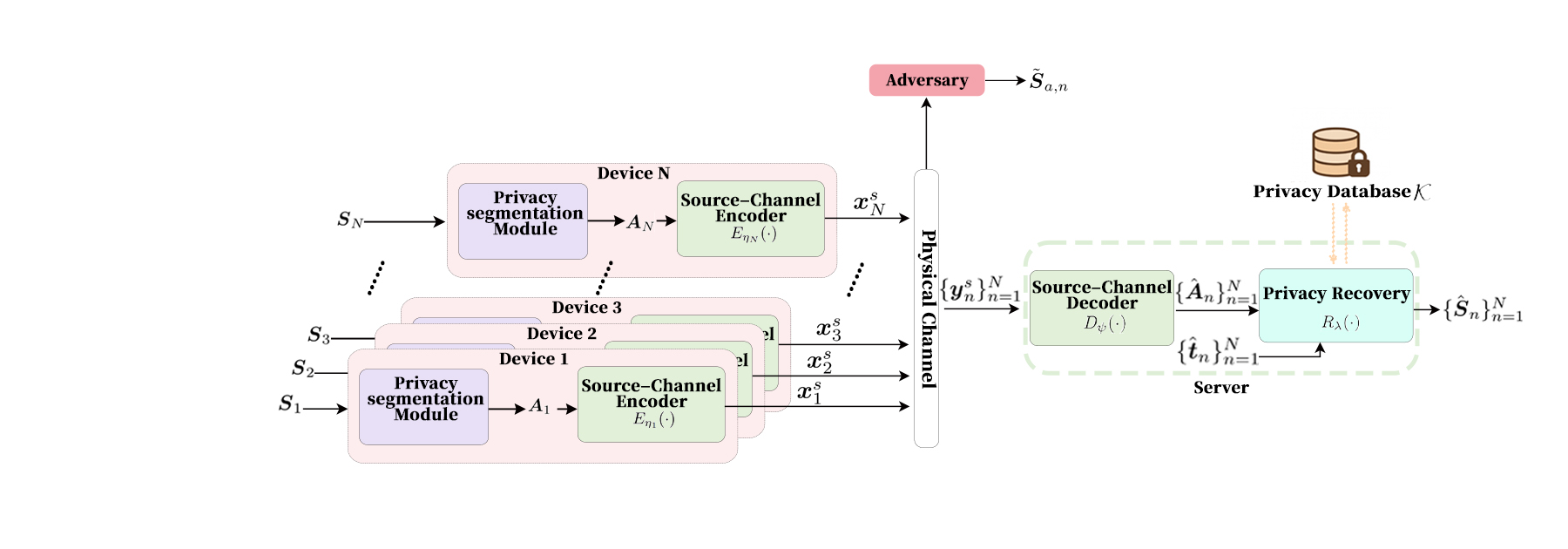}
    }
    \vspace{-1.2cm}
    \caption{Privacy protection module to remove sensitive information  prior to transmission and recover the image at the server using a vision-language model.}
    \label{fig:system_model_image}
\end{figure*}

Upon detection of sensitive content, each device  applies a segmentation module to separate the privacy region from the background in the captured image and  remove the corresponding area accordingly: 
\begin{equation}
\boldsymbol{A}_n = f_{\mathrm{SEG}}(\boldsymbol{S}_n| \hat{i}),
\end{equation}
where $f_{\mathrm{SEG}}(\cdot)$ denotes a segmentation function \cite{kirillov2023segment} that generates a mask to suppress privacy-sensitive objects. 
Here, we consider a single-object scenario, where  $f_{\mathrm{SEG}}(\cdot)$  can be realized by a foreground–background segmentation, with the sensitive object treated as the foreground region and the background considered as non-sensitive\footnote{The proposed framework can be  extended to multi-object scenarios by decomposing a multi-object image into several single-object regions, and then the same privacy procedure is applied to mask sensitive objects individually. Finally,  all processed regions are merged to reconstruct the sanitized image.}. 
The output $\boldsymbol{A}_n \in \mathbb{R}^{C \times H \times W}$ is a sanitized image that masks sensitive content. 

As shown in Fig.~\ref{fig:system_model_image}, the sanitized image $\boldsymbol{A}_n$ is then encoded by a source-channel encoder $E_{\eta_n}(\cdot)$ into a transmit signal: 
\begin{equation}
\boldsymbol{x}_n^{s} = E_{\eta_n}(\boldsymbol{A}_n),
\end{equation}
where $\eta_n$ is the trainable parameter of device $n$ and $\boldsymbol{\eta} = [\eta_1, \eta_2, \ldots, \eta_N]$ collects the parameters of all devices. The received signal at the server is
\begin{equation}
    \boldsymbol{y}_n^{s}
    =
    \boldsymbol{h}_n
    \odot
    \boldsymbol{x}_n^{s}
    +
    \boldsymbol{n}.
\label{ys}
\end{equation}

Upon signal reception, the server applies a decoder $D_{\psi}(\cdot)$ with parameter $\psi$ to reconstruct the masked image $\hat{\boldsymbol{A}}_n\in \mathbb{R}^{C \times H \times W}$, via
\begin{equation}
    \hat{\boldsymbol{A}}_n = D_{\psi}( \boldsymbol{y}_n^{s} ).
\end{equation} 
The decoded output are then passed to a privacy recovery module $R_{\lambda}(\cdot)$, based on a vision–language model (VLM) with parameter $\lambda$, to recover the complete image as 
\begin{equation}
\hat{\boldsymbol{S}}_n = R_\lambda(\hat{\boldsymbol{A}}_n | 
\hat{i},
\hat{\boldsymbol{{t}}}_n,
\mathcal{K}).
\label{equ_recovery}
\end{equation}

\subsection{Adversary Model} 
We consider a model-aware adversary that has full access to the model parameters deployed on the edge devices. This threat model is well motivated in wireless edge systems, where the limited computational and security capabilities of edge devices make them more vulnerable to model leakage and compromise.  
In contrast, 
the server-side model parameters and privacy dataset $\mathcal{K}$ are assumed to be inaccessible,  as the edge server provides a substantially higher level of physical and cyber security. 
By exploiting the exposed edge models, the adversary can launch model inversion attacks to infer semantic information. 
We further assume that all edge devices share the same backbone architecture for the text encoder $T_{\theta_n}$ and  image encoder $E_{\eta_n}$, a common practice in edge computing platforms to reduce deployment and maintenance cost.
Thus, compromising a single edge device enables the adversary to recover semantic representations through model inversion. 

Specifically, the adversary collects textual samples $\boldsymbol{t}_a = [w_{a,1},   \ldots, w_{a,L}]$ from public datasets  and feeds them into the encoder $T_{\theta_n}(\cdot)$ to obtain the corresponding outputs, forming a large number of input–output pairs. 
Based on these pairs, the adversary then trains a surrogate decoder $D_{\beta_a}(\cdot)$ to minimize the cross-entropy  loss: 
$\mathcal{D}_{\mathrm{CE}}(p,q)
= -\sum_{l=1}^{L} 
p(w_l)\log\big(q(w_l)\big)$, 
where $p(w_l)$ denotes the ground-truth token probability induced by the input text  $\boldsymbol{t}$, and $q(w_l)$ denotes the predicted token probability produced by the decoder. Accordingly, the loss function is given as follows:
\begin{equation}
\mathcal{L}_{A}({\beta}_a)
= \mathcal{D}_{\mathrm{CE}}\!\left(
\boldsymbol{t}_a,\,
D_{{\beta}_a}\!\left(
T_{\theta_n}(\boldsymbol{t}_a)
\right)
\right).
\label{A}
\end{equation} 
Similarly, the adversary can intercept the privacy-removed image representations. 
Using the same model inversion strategy, the adversary feeds image samples from public datasets to the encoder $E_{\eta_n}(\cdot)$ to obtain latent representations, and train a surrogate image decoder $D_{\psi_a}(\cdot)$ to reconstruct the original visual content. 
In addition, the adversary employs a VLM module $R_{\lambda_a}(\cdot)$ with capability comparable to that of the server, but whose parameters are derived from a publicly available VLM without access to the privacy dataset. Consequently, the image recovered by the adversary is  
\begin{equation}
\tilde{\boldsymbol{S}}_{a,n}
= R_{\lambda_a}\!\left(
D_{\psi_a}\!\left(\boldsymbol{{y}}_{a,n}^{s}\right) |
D_{{\beta}_a}\!\left(\boldsymbol{{y}}_{a,n}^{t}\right)
\right).
\label{equ_adv}
\end{equation} 
However, without access to the privacy database $\mathcal{K}$, the adversary is limited to partial image reconstruction, with sensitive content either omitted or incorrectly inferred.

\subsection{Semantic Information Bottleneck} \label{DistributedInformationBottleneck} 

Under the  multi-view setting, different edge devices may capture either complementary or redundant semantic information from diverse viewpoints. 
Although complementary semantic cues improve downstream privacy identification and image reconstruction, redundant cross-view information reduces  communication efficiency. 
To address this issue, a multi-device semantic information bottleneck (SIB) loss is introduced to preserve task-relevant semantic information while suppressing redundant cross-view correlations. Accordingly, the SIB loss is formulated as
\begin{equation} 
\mathcal{L}_{\mathrm{SIB}}(\beta,\boldsymbol{\Theta})
=
H(i | \boldsymbol{y}_{1:N}^{t})
+
\gamma
\sum_{n=1}^{N}
\sum_{m=n+1}^{N}
I(\boldsymbol{y}_{n}^{t};\boldsymbol{y}_{m}^{t}), 
\label{sib}
\end{equation}
where $H(i | \boldsymbol{y}_{1:N}^{t})$ represents the conditional entropy of the ground-truth  identity label $i$, 
given all  text messages. 
It characterizes the uncertainty of privacy identification after multi-message fusion, and its minimization encourages   $\boldsymbol{y}_{1:N}^{t}$ to retain sufficient  information to tell accurate  identification.
Meanwhile,  $\sum_{n} \sum_{m} I(\boldsymbol{y}_{n}^{t};\boldsymbol{y}_{m}^{t})$
measures the  mutual  semantic information  across different devices' viewpoints. 
Minimizing this term suppresses   semantic redundancy and promotes complementary  representations across devices.  
Finally,  $\gamma >0$ controls the tradeoff between preserving task-relevant semantic information and reducing cross-view redundancy. 


\subsection{Problem Formulation} \label{problemFormulation}  


The goal of the proposed privacy-preserving semantic communication framework is to minimize the image reconstruction distortion at the server, while limiting privacy leakage to the adversary and reducing redundant cross-view semantic information in the  multi-device communication system. 
In the first stage of privacy detection, the accuracy of the reconstructed text $\hat{\boldsymbol{t}}_n$ and the privacy identity label $\hat{i}$ will directly affect the reconstruction performance of $\boldsymbol{\hat{S}}_n$ in the second stage, as shown in  (\ref{equ_recovery}).  
Therefore, the reconstruction performance at the server can be quantified as
\begin{equation}
    O_{s} = \frac{1}{N}\sum_{n=1}^N  \left\|\boldsymbol{S}_n - \hat{\boldsymbol{S}}_n(\beta,\phi,\psi,\lambda,\theta_n,\eta_n)\right\|_2^2.
\label{O_s}
\end{equation}
Meanwhile, according to (\ref{equ_adv}), the performance of the privacy detection module also affects the adversary's reconstruction quality of $\tilde{\boldsymbol{S}}_{a,n}$, which can be characterized by 
\begin{equation}
    O_a = \frac{1}{N}\sum_{n=1}^N  \left\|\boldsymbol{S}_n - \tilde{\boldsymbol{S}}_{a,n}(\beta,\phi,\theta_n,\eta_n)\right\|_2^2.
\end{equation}
Therefore, the overall objective function is 
\begin{equation}
    O  = O_s(\beta,\phi,\psi,\lambda,\boldsymbol{\Theta},\boldsymbol{\eta}) - O_a(\beta,\phi,\boldsymbol{\Theta},\boldsymbol{\eta}) +  \mu \mathcal{L}_{\mathrm{SIB}}  (\beta, \boldsymbol{\Theta}),
\label{O}
\end{equation} 
where $\mu$ is a regularization weight. 
Consequently, the privacy-preserving semantic communication problem can be formulated as: 
\begin{subequations}\label{equs_opt}  
	\begin{align}
		\min_{ \beta,\phi,\psi,\lambda,\boldsymbol{\Theta},\boldsymbol{\eta}  } \quad &       O (\beta,\phi,\psi,\lambda,\boldsymbol{\Theta},\boldsymbol{\eta}) & \label{equ_obj1}\\ 
		\textrm{s. t.} \quad  
        & \frac{\left\| \boldsymbol{h}_n \boldsymbol{x}_n^{t}(\theta_n) \right\|^{2} }{B\,\sigma_n^{2}} \ge \tau_{\mathrm{thre}}^{(t)}, & \forall n, \label{cons_snr_t} \\
		& \left\| \boldsymbol{x}_n^{t}(\theta_n) \right\|^{2} \le P_{\max}^{(t)}, & \forall n, \label{cons_power_t} \\
		& \frac{\left\| \boldsymbol{h}_n \boldsymbol{x}_n^{s}(\beta,\phi,\theta_n,\eta_n) \right\|^{2} }{ B\,\sigma_n^{2} } \ge \tau_{\mathrm{thre}}^{(s)}, & \forall n, \label{cons_snr_i} \\
        & \left\| \boldsymbol{x}_n^{s}(\beta,\phi,\theta_n,\eta_n) \right\|^{2} \le P_{\max}^{(s)}, & \forall n, \label{cons_power_i} \\ 
        & \left\| \boldsymbol{k}^{\mathrm{dev}}_n - \boldsymbol{k}^{\mathrm{srv}}_n \right\|^{2} \le \epsilon_k, & \forall n,  \label{cons_key}
	\end{align} 
\end{subequations}
where constraints (\ref{cons_snr_t}) and (\ref{cons_snr_i})  ensure that the received signal-to-noise ratio (SNR) exceed the predefined thresholds $\tau_{\mathrm{thre}}^{(t)}$ and $\tau_{\mathrm{thre}}^{(s)}$,  
constraints (\ref{cons_power_t}) and ((\ref{cons_power_i})) limit the transmission powers of the two-stage transmission by  $P_{\max}^{(t)}$ and $P_{\max}^{(s)}$, respectively, 
and constraint (\ref{cons_key}) ensures that the key discrepancy between  each device $n$ and the server  in the physical-layer encryption stage is bounded   within  $\epsilon_k$.

The problem  (\ref{equs_opt}) is challenging to solve for three reasons.
Firstly, the  multi-view wireless setting introduces a tradeoff between semantic fidelity and communication efficiency. Since different devices capture the same object from different viewpoints, the transmitted representations contain both complementary and redundant semantic information. With the information bottleneck imposed by the wireless channel on semantic transmission, the solution must carefully preserve task-relevant information while suppressing  cross-view redundancy. 
Secondly, the multi-modal nature of the framework requires textual descriptions and visual information to jointly guide privacy identification and image reconstruction. Therefore, semantic consistency across different modalities must be preserved to ensure reliable recovery performance. 
Thirdly, the presence of a model-aware adversary complicates the framework design. The system must maintain sufficient semantic information for accurate reconstruction at the server, while simultaneously minimizing privacy leakage caused by intercepted textual signals, privacy-removed images, and exposed edge-side model parameters.


\section{Solution} \label{solution}


Although the system-level optimization problem  (\ref{equs_opt}) is  difficult to solve across two stages in an end-to-end manner, the optimization variables are largely decoupled within the constraints. 
Specifically, the privacy detection and protection modules, as well as the transmitter and receiver, operate sequentially,  resulting in relatively independent optimization processes.
Therefore, (\ref{equs_opt}) can be decomposed into several subproblems and optimized on a module-wise basis. 
An iterative algorithm is then developed to alternately solve the resulting subproblems until convergence. 

\subsection{Privacy Detection Module}

\subsubsection{Encrypted Semantic-Channel Text Encoder and Decoder}

We first optimize the encryption semantic-channel text (ESCT) encoder and  decoder in the privacy detection stage, as shown in Fig. \ref{fig:system_model_text}. 
Here, the aim is to improve the  reconstructed text quality, thus enhancing the server-side reconstruction objective in  (\ref{O_s}) for the next stage. 

Specifically, the encoder parameters $\boldsymbol{\Theta}$  and decoder parameter $\beta$ are optimized to minimize the discrepancy between the original text $\boldsymbol{t}_n$ and the reconstructed text $\hat{\boldsymbol{t}}_n$. The corresponding loss function is given by 
\begin{equation}
\mathcal{L}_K( \beta, \boldsymbol{\Theta}) = \frac{1}{N} \sum_{n=1}^{N} \mathcal{D}_{\mathrm{CE}} \left(\boldsymbol{t}_n, \hat{\boldsymbol{t}}_n \right).
\end{equation}
Meanwhile,  the adversary aims to maximize privacy leakage, by minimizing the adversarial loss in (\ref{A}). 
To optimize the decoder $D_{\beta_a}(\cdot)$ and reconstruct the original text from the intercepted signal $\boldsymbol{y}_{a,n}$, the adversarial reconstruction loss is 
\begin{equation}
\mathcal{L}_{\mathrm{adv}}(\beta_a | \boldsymbol{\Theta})
=
\frac{1}{N}
\sum_{n=1}^{N}
\mathcal{D}_{\mathrm{CE}}
\left(
\boldsymbol{t}_n,\,
D_{\beta_a}(\boldsymbol{y}_{a,n})
\right).
\end{equation}
Consequently, the privacy detection stage can be reformulated as an adversarial semantic-channel training problem. 
Rather than directly optimizing the image reconstruction quality in (\ref{equ_obj1}), this stage first optimizes the privacy detection modules to improve the quality of the reconstructed text. The resulting loss function is given by 
\begin{equation} \label{L_D}
\mathcal{L}_D(\boldsymbol{\Theta},\beta)
=
\mathcal{L}_K(\boldsymbol{\Theta},\beta)
-
\mathcal{L}_{\mathrm{adv}}(\boldsymbol{\Theta})
+  
\mu \mathcal{L}_{\mathrm{SIB}}  (\boldsymbol{\Theta}, \beta).
\end{equation} 
which optimizes the encrypted semantic-channel transceiver jointly with the  semantic information bottleneck.  


\begin{figure}[t] 
    \centering
    \hspace*{-0.65cm}
    \includegraphics[width=0.54\textwidth]{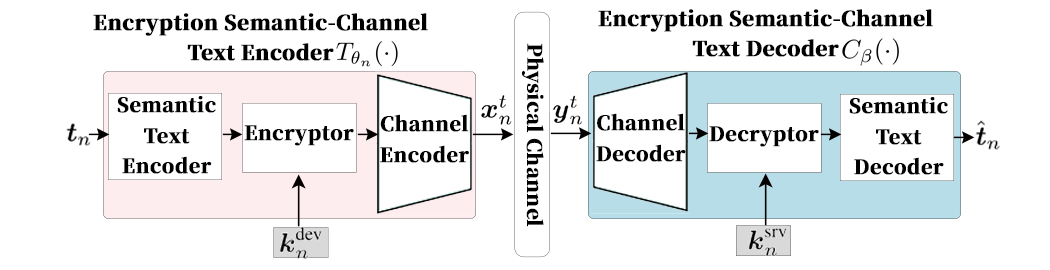}\vspace{-0.1cm}
    \caption{\label{fig_pr}Encryption semantic-channel encoder and decoder for text message in the privacy detection stage.} 
    \label{fig:ESCED}
    \vspace{-0.1cm}
\end{figure}

\begin{figure}[t]  
    \centering
    \includegraphics[width=0.47\textwidth]{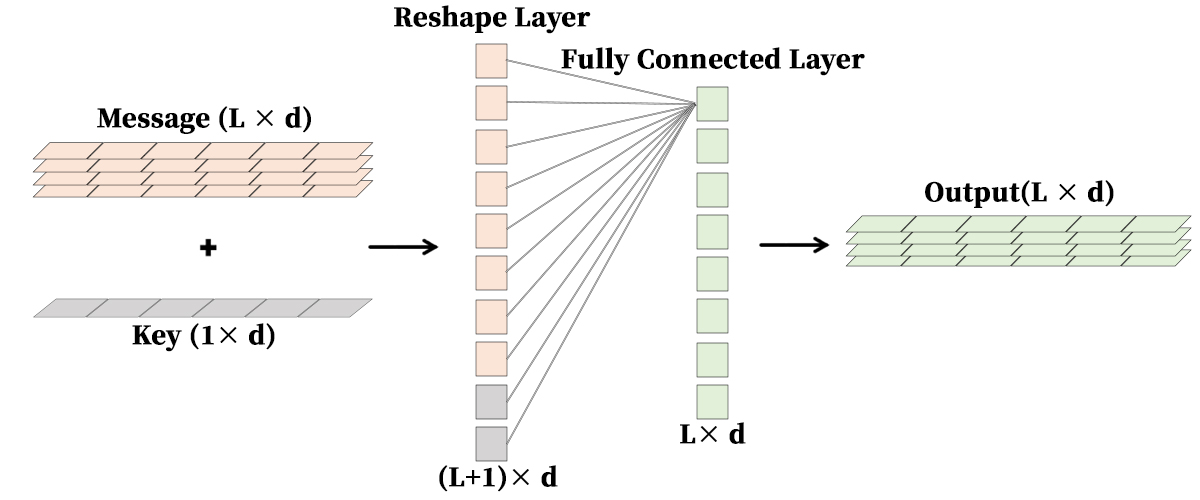}\vspace{-0.1cm}
    \caption{\label{fig_pr}Encryptor architecture that fuses semantic features and physical-layer secret keys through a fully connected layer.} 
    \label{fig:encryption}
\end{figure}

The overall architecture of the ESCT encoder and decoder is given in Fig. \ref{fig:ESCED}, where a transformer-based network is employed for semantic coding and decoding, with an autoencoder-based architecture  for channel coding. 
By jointly optimizing semantic and channel coding, the proposed framework reduces the  transmitted data amount while preserving the semantic meaning, to improve communication  efficiency and robustness, particularly in low signal-to-noise 
ratio  (SNR) environments  \cite{deepsc-2021}.  
Furthermore, to defend against model inversion attacks and protect semantic text information, physical-layer encryption and decryption modules are introduced. 
As shown in Figs. \ref{fig:ESCED} and \ref{fig:encryption}, a encryptor takes the secret key $\boldsymbol{k}_n^{\mathrm{dev}}$ and the semantic encoder output. 
After concatenation, the combined representation is flattened, transformed through a reshape layer, and projected back to the original embedding dimension. 
Correspondingly, the decryptor utilizes the channel decoder output and the secret key $\boldsymbol{k}_n^\text{dev}$   to reconstruct the textual information $\hat{\boldsymbol{{t}}}_n$.  

The training procedure for the ESCT transceiver is summarized in Algorithm~\ref{alg1}.

\begin{algorithm}[t]
\caption{\label{alg1}Train the Encrypted Text Semantic Transceiver}
\begin{algorithmic}[1]\small

\State \textbf{Input}: 
$\mathcal{S}$,
$D_{\beta_a}(\cdot)$,
$\tau_{\mathrm{thre}}^{(t)}$,
$P_{\max}^{(t)}$,
and $\epsilon_k$.

\While{$\mathcal{L}_D$ in (\ref{L_D})   has not converged}

        \State Sample a multi-view image set $\{\boldsymbol{S}_n\}_{n=1}^{N}$ from $\mathcal{S}$.

        \For{$n=1,\ldots,N$}
        \State \textbf{Channel Probing and Key Generation:}
            \State \hspace{0.5cm}
            Device $n$ and the server exchange pilot signals.
            \State \hspace{0.5cm}
            Estimate $\boldsymbol{h}^{\mathrm{dev}}_n$ and $\boldsymbol{h}^{\mathrm{srv}}_n$.
            \State \hspace{0.5cm}
            $f_{\mathrm{PLK}}(\boldsymbol{h}^{\mathrm{dev}}_n|\epsilon_k)\rightarrow \boldsymbol{k}^{\mathrm{dev}}_n$.
            \State \hspace{0.5cm}
            $f_{\mathrm{PLK}}(\boldsymbol{h}^{\mathrm{srv}}_n |\epsilon_k)\rightarrow \boldsymbol{k}^{\mathrm{srv}}_n$.

            \State \textbf{Transmitter:}
        
            \State \hspace{0.5cm} $f_{\mathrm{BLIP}}(\boldsymbol{S}_n)\rightarrow \boldsymbol{t}_n$ 
        
            \State \hspace{0.5cm}
            $T_{\boldsymbol{\theta}_n}
            (\boldsymbol{t}_n,\boldsymbol{k}^{\mathrm{dev}}_n| \tau_{\mathrm{thre}}^{(t)},
            P_{\max}^{(t)})
            \rightarrow
            \boldsymbol{x}_n^{t}$

            \State \hspace{0.5cm} Transmit $\boldsymbol{x}_n^{t}$ over the channel.
        
            \State \textbf{Receiver:}
        
            \State \hspace{0.5cm} Receive $\boldsymbol{y}_n^{t}$ via Eq. (\ref{equ:x2yt}).
        
            \State \hspace{0.5cm}
            $C_{\boldsymbol{\beta}}
            (\boldsymbol{y}_n^{t},\boldsymbol{k}^{\mathrm{srv}}_n| \tau_{\mathrm{thre}}^{(t)},
            P_{\max}^{(t)})
            \rightarrow
            \hat{\boldsymbol{t}}_n$.
        
            \State \textbf{Adversary:}
        
            \State \hspace{0.5cm}
            $D_{\boldsymbol{\beta}_a}
            (\boldsymbol{y}_{a,n}^{t})
            \rightarrow
            \hat{\boldsymbol{t}}_{a,n}$.
        
        \EndFor
        
        \State Compute the loss $\mathcal{L}_{D}$ by (\ref{L_D}).
        
        \State Update $\boldsymbol{\Theta}$ and $\boldsymbol{\beta}$ using gradient descent.

\EndWhile

\State \textbf{Output}: $\{T_{\boldsymbol{\theta}_n}(\cdot)\}_{n=1}^{N}$ and $C_{\boldsymbol{\beta}}(\cdot)$.
 
\end{algorithmic}
\end{algorithm}

\subsubsection{Text-based Privacy Identification Module}

Given the trained ESCT encoder and  decoder, the next component in Fig. \ref{fig:system_model_text}  is the privacy identification module $I_{\phi}(\cdot)$, which consists of a text encoder $G_{\phi}(\cdot)$ and a feature matching module. 

First, the reconstructed text  messages
$\{\hat{\boldsymbol{t}}_n\}_{n=1}^{N}$
from all devices  are jointly processed by 
\begin{equation}
    \{\hat{\boldsymbol{f}}_n \}_{n=1}^{N} = G_{\phi}( \{ \hat{\boldsymbol{t}}_n \}_{n=1}^{N} ),
\label{textencoder}
\end{equation} 
where $\hat{\boldsymbol{f}}_n\in \mathbb{R}^d$ is the identity-discriminative feature  extracted by the text encoder $G_{\phi}(\cdot)$. 
Meanwhile, each image $\{ \boldsymbol{S}^{*}_{i,m} \}_{m=1}^{M}$ in the privacy database is  converted into textual descriptions, using the same process as (\ref{equ:s2t}), which is subsequently encoded by the same encoder in (\ref{textencoder}) to obtain the corresponding feature representation   $\{ \boldsymbol{f}^{*}_{i,m} \}_{m=1}^{M}$. 

To determine the identity revealed in the text message, the query feature  $\{ \hat{\boldsymbol{f}}_n \}_{n=1}^{N}$ is matched against each reference feature  $\{ \boldsymbol{f}^{*}_{i,m} \}_{m=1}^{M}$ in a set-to-set manner. 
The pairwise similarity between the $n$-th query feature and the $m$-th reference feature of entity $i$ is computed as
\begin{equation}
\zeta_{n,m}^{(i)} = \hat{\boldsymbol{f}}_n^H \boldsymbol{f}^{*}_{i,m}. 
\end{equation}
Since reference images in the privacy database  originate from different viewpoints, it is unknown which one best matches the query feature of device $n$. Therefore, the maximum similarity over all reference features is selected for each query, and 
the matching score of entity $i$ is the average of the maximum similarities of all viewpoints. 
The identity with the highest matching score is selected as the output, i.e.,
\begin{equation}
\hat{i} = \arg\max_i \frac{1}{N} \sum_{n=1}^{N} \max_{m}  \zeta_{n,m}^{(i)}.
\end{equation}

To training $G_{\phi}(\cdot)$ for identity representations, contrastive learning is adopted to  cluster features of the same identity while separating those of different identities. 
For an anchor text message $\boldsymbol{t}_i$, the contrastive loss is defined as
\begin{equation}
\mathcal{L}_C(\boldsymbol{\phi}) = - \log \frac
    { 
        \exp \left(
        G_{\phi}(\boldsymbol{t}_i) 
        \cdot 
        G_{\phi}(\boldsymbol{t}_i^{+}) 
        / \tau \right)
    }
    {
        \sum_{j}\exp \left(
        G_{\phi}(\boldsymbol{t}_i)
        \cdot
        G_{\phi}(\boldsymbol{t}_j)
        / \tau \right)
    },
\label{L_C}
\end{equation}
where $\boldsymbol{t}_i^{+}$ denotes a positive sample of the same identity $i$, $\boldsymbol{t}_j$ represents all samples in the batch, and $\tau$ is a temperature parameter. 
Minimizing (\ref{L_C}) promotes intra-class compactness and inter-class separability in the learned embedding space. 
Here, the anchor text message  is generated from training images in $\mathcal{S}$ 
rather than from the privacy database $\mathcal{K}$. 
The corresponding training procedure is summarized in Algorithm~\ref{alg_reid}.

\begin{algorithm}[t]
\caption{\label{alg_reid}Train the Privacy Identification Module}
\begin{algorithmic}[1]\small

\State \textbf{Input}: $\mathcal{S}$.

\While{$\mathcal{L}_{C}$ in (\ref{L_C}) has not converged}

        \State Sample a mini-batch of images with identity labels from $\mathcal{S}$.

        \For{each image}
            \State Generate textual description:
            \State \hspace{0.5cm} $f_{\mathrm{BLIP}}(\boldsymbol{S}) \rightarrow \boldsymbol{t}$.
        
            \State Extract identity-discriminative feature:
            \State \hspace{0.5cm} $G_{\phi}(\boldsymbol{t}) \rightarrow \boldsymbol{f}$.
        \EndFor
    
        \State Compute the contrastive loss $\mathcal{L}_{C}$ by (\ref{L_C}).
    
        \State Update $\phi$ using gradient descent.

\EndWhile

\State \textbf{Output}: $G_{\phi}(\cdot)$.

\end{algorithmic}
\end{algorithm}

\subsubsection{Approximation  of multi-device SIB Loss}

A practical challenge of the proposed framework is that the SIB loss in (\ref{sib}), as well as the loss in (\ref{L_D}), contains conditional entropy and mutual information terms that are generally intractable to compute, making the optimization in Algorithm \ref{alg1} difficult.  
To enable efficient training and improve convergence, we adopt tractable approximations for these quantities. 

The first term of the conditional entropy is defined as $H(i | \boldsymbol{y}_{1:N}^{t}) = -\mathbb{E}_{p(i,\boldsymbol{y}_{1:N}^{t})}[\log p(i | \boldsymbol{y}_{1:N}^{t})]$ 
which is generally intractable, because the true posterior distribution  $p(i | \boldsymbol{y}_{1:N}^{t})$ is  unknown. 
To obtain a tractable surrogate objective, we introduce a variational posterior distribution  $q_{\phi}(i | \boldsymbol{y}_{1:N}^{t})$, parameterized by the text-based privacy identification module with parameter $\phi$,  to approximate the true posterior.  
Using the non-negativity property of the KL divergence, we first have $\mathbb{E}_{p(\boldsymbol{y}_{1:N}^{t})} \left[ D_{\mathrm{KL}} \left( p(i | \boldsymbol{y}_{1:N}^{t}) \parallel q_{\phi}(i | \boldsymbol{y}_{1:N}^{t}) \right) \right] \ge 0$, 
which yields  
\begin{equation}
-\mathbb{E}_{p(i,\boldsymbol{y}_{1:N}^{t})} \left[ \log q_{\phi} \left( i | \boldsymbol{y}_{1:N}^{t} \right) \right] - H(i | \boldsymbol{y}_{1:N}^{t}) \ge 0. 
\end{equation}
Therefore, the conditional entropy can be upper-bounded as 
\begin{equation}
H(i | \boldsymbol{y}_{1:N}^{t}) \le 
-\mathbb{E}_{p(i,\boldsymbol{y}_{1:N}^{t})}
\left[
\log q_{\phi}
\left(
i | \boldsymbol{y}_{1:N}^{t}
\right)
\right]
\triangleq \mathcal{L}_{\mathrm{task}}.
\label{L_task}
\end{equation}
Here, minimizing  $\mathcal{L}_{\mathrm{task}}$  encourages the variational posterior $q_{\phi}(i | \boldsymbol{y}_{1:N}^{t})$ to accurately predict the privacy identity $i$ from the reconstructed semantic representations, thereby reducing the uncertainty of privacy identification and indirectly minimizing the conditional entropy term. 
In practice, $q_{\phi}(i | \boldsymbol{y}_{1:N}^{t})$ corresponds to the output probability distribution of the text-based privacy identification network. 
Consequently, $\mathcal{L}_{\mathrm{task}}$  reduces to the standard cross-entropy classification loss, which can be efficiently optimized using stochastic gradient descent.

The second term of SIB loss is the mutual information 
\begin{equation}
    I(\boldsymbol{y}_n^{t};\boldsymbol{y}_m^{t})=
\mathbb{E}_{p(\boldsymbol{y}_n^{t},\boldsymbol{y}_m^{t})}
\left[
\log
\frac{
p(\boldsymbol{y}_n^{t},\boldsymbol{y}_m^{t})
}{
p(\boldsymbol{y}_n^{t})
p(\boldsymbol{y}_m^{t})
}
\right], 
\label{equ:mutual_info}
\end{equation} 
which measures the amount of shared semantic information extracted from  different views of two edge devices $n$ and $m$.
However, directly computing this quantity is generally intractable since the joint distribution $p(\boldsymbol{y}_n^{t},\boldsymbol{y}_m^{t})$, as well as the marginal distributions $p(\boldsymbol{y}_n^{t})$ and $p(\boldsymbol{y}_m^{t})$ are unknown. 
To obtain a tractable approximation, we adopt the variational Contrastive Log-ratio Upper Bound (vCLUB) estimator \cite{cheng2020club_arxiv} and introduce a variational conditional distribution $q_{\omega} (\boldsymbol{y}_m^{t}|\boldsymbol{y}_n^{t})$, parameterized by $\omega$, to approximate the dependency between the semantic representations from different views. 
In particular, $q_{\omega}$ is trained using positive and negative multi-view pairs, where given a mini-batch with $B$ objects, for the representation  $(\boldsymbol{y}_n^{t})^{(b)}$, a positive pair $ (\boldsymbol{y}_m^{t})^{(b)}$ is selected from a different  view  $m$ of the same object $b$, 
while a negative sample $(\boldsymbol{y}_m^{t})^{(b')}$
is selected from a different view $m$ of a different object $b'\neq b$. 
The corresponding training loss is given by
\begin{equation}
\begin{aligned}
\mathcal{L}_{\mathrm{vCLUB}}(\omega) = \frac{1}{B} \sum_{b=1}^{B} \Big[ 
& \log q_{\omega} \left( (\boldsymbol{y}_m^{t})^{(b)} | (\boldsymbol{y}_n^{t})^{(b)} \right) \\
& - \log q_{\omega} \left( (\boldsymbol{y}_m^{t})^{(b')} | (\boldsymbol{y}_n^{t})^{(b)} \right)
\Big],
\end{aligned}
\label{vclub}
\end{equation}
This objective encourages $q_{\omega}$ to assign a high likelihood to semantic representations corresponding to the same object while assigning a low likelihood to representations originating from different objects. Consequently, the learned conditional distribution captures the statistical dependency between correlated semantic representations.

Using the trained  conditional distribution $q_{\omega}(\boldsymbol{y}_m^{t}|\boldsymbol{y}_n^{t})$, the vCLUB estimator of mutual information \cite{cheng2020club_arxiv}  is given by 
\begin{equation}
\begin{aligned}
I_{\mathrm{vCLUB}}
(\boldsymbol{y}_n^{t};\boldsymbol{y}_m^{t})
=
&\,
\mathbb{E}_{p(\boldsymbol{y}_n^{t},\boldsymbol{y}_m^{t})}
\bigg[
\log q_{\omega}
(\boldsymbol{y}_m^{t}|\boldsymbol{y}_n^{t})
\bigg]
\\
&
-
\mathbb{E}_{p(\boldsymbol{y}_n^{t})p(\boldsymbol{y}_m^{t})}
\bigg[
\log q_{\omega}
(\boldsymbol{y}_m^{t}|\boldsymbol{y}_n^{t})
\bigg].
\end{aligned}
\label{eq:vclub_bound}
\end{equation} 
The following theorem shows that the estimation in (\ref{eq:vclub_bound}) provides an upper bound on the true mutual information in (\ref{equ:mutual_info}), under the condition that the variational conditional distribution $q_{\omega}$ must be well trained, i.e., $q_{\omega} (\boldsymbol{y}_m^{t}|\boldsymbol{y}_n^{t}) \approx p(\boldsymbol{y}_m^{t}|\boldsymbol{y}_n^{t}) $. 
\begin{theorem} \label{theorem1} 
Given $D_\mathrm{KL} (p(\boldsymbol{y}_n^{t},\boldsymbol{y}_m^{t}) ||
q_{\omega} (\boldsymbol{y}_m^{t}|\boldsymbol{y}_n^{t}) p(\boldsymbol{y}_n^{t}) ) \le D_\mathrm{KL} ( p(\boldsymbol{y}_n^{t}) p(\boldsymbol{y}_m^{t}) || q_{\omega} (\boldsymbol{y}_m^{t}|\boldsymbol{y}_n^{t}) p(\boldsymbol{y}_n^{t}))$ holds, 
the mutual information is upper-bounded by the vCLUB estimator, i.e., 
\begin{equation}
I(\boldsymbol{y}_n^{t};\boldsymbol{y}_m^{t})
\le
I_{\mathrm{vCLUB}}
(\boldsymbol{y}_n^{t};\boldsymbol{y}_m^{t}). 
\end{equation}
\end{theorem}
\begin{proof}
For notational simplicity, let
$\boldsymbol{u}=\boldsymbol{y}_n^{t}$ and
$\boldsymbol{v}=\boldsymbol{y}_m^{t}$. 
The gap between  vCLUB and the true mutual
information is  
\begin{align}
\Delta
&=I_{\mathrm{vCLUB}}(\boldsymbol{u};\boldsymbol{v}) -I(\boldsymbol{u};\boldsymbol{v})  \notag \\
&=\mathbb{E}_{p(\boldsymbol{u},\boldsymbol{v})}\left[\log q_{\omega}(\boldsymbol{v}|\boldsymbol{u})\right] -\mathbb{E}_{p(\boldsymbol{u})p(\boldsymbol{v})}\left[\log q_{\omega}(\boldsymbol{v}|\boldsymbol{u})\right] \notag \\&\quad-\mathbb{E}_{p(\boldsymbol{u},\boldsymbol{v})}\left[\log p(\boldsymbol{v}|\boldsymbol{u})-\log p(\boldsymbol{v})\right] \notag \\
&=\mathbb{E}_{p(\boldsymbol{u})p(\boldsymbol{v})}\left[\log\frac{p(\boldsymbol{v})}{q_{\omega}(\boldsymbol{v}|\boldsymbol{u})}\right]-\mathbb{E}_{p(\boldsymbol{u},\boldsymbol{v})}\left[\log\frac{p(\boldsymbol{v}|\boldsymbol{u})}{q_{\omega}(\boldsymbol{v}|\boldsymbol{u})}\right] \notag \\
&=D_{\mathrm{KL}}\left(p(\boldsymbol{u})p(\boldsymbol{v})\|q_{\omega}(\boldsymbol{v}|\boldsymbol{u})p(\boldsymbol{u})\right)  \notag \\&\quad-D_{\mathrm{KL}}\left(p(\boldsymbol{u},\boldsymbol{v})\|q_{\omega}(\boldsymbol{v}|\boldsymbol{u})p(\boldsymbol{u})\right).
\end{align} 
According to the condition, the second KL-divergence term is no larger than the first.
Therefore, $\Delta \ge 0$, which implies
\begin{equation}
I(\boldsymbol{y}_n^{t};\boldsymbol{y}_m^{t})\le I_{\mathrm{vCLUB}}(\boldsymbol{y}_n^{t};\boldsymbol{y}_m^{t}).
\end{equation}
This completes the proof.
\end{proof}

Based on Theorem \ref{theorem1}, the cross-view redundancy reduction loss is approximated as
\begin{equation} \label{L_red} 
\sum_{n=1}^{N} \sum_{m=n+1}^{N} I (\boldsymbol{y}_n^{t};\boldsymbol{y}_m^{t}) \le \sum_{n=1}^{N} \sum_{m=n+1}^{N} I_{\mathrm{vCLUB}} (\boldsymbol{y}_n^{t};\boldsymbol{y}_m^{t})  \triangleq \mathcal{L}_{\mathrm{red}}.  
\end{equation}
By combining (\ref{L_task}) and (\ref{L_red}), the adversarial semantic-channel training objective in (\ref{L_D}) can be rewritten as
\begin{equation}\label{L_D_2}
\begin{aligned}
\mathcal{L}_D
\le 
\bar{\mathcal{L}}_D(\boldsymbol{\Theta},\beta) = 
\mathcal{L}_K(\boldsymbol{\Theta},\beta)
-
\mathcal{L}_{\mathrm{adv}}(\boldsymbol{\Theta})
\\
&\hspace{-3.6cm} 
+  
\mu ( \mathcal{L}_{\mathrm{task}}(\boldsymbol{\Theta}) + \gamma \mathcal{L}_{\mathrm{red}}(\boldsymbol{\Theta}) ).
\end{aligned}
\end{equation} 
For a practical learning, Algorithm \ref{alg1} will use the new loss function in (\ref{L_D_2}), instead of (\ref{L_D}) for the parameter training, which minimizes the upper bound of the original loss.


\subsection{Privacy Protection Module}

Given that all components in the privacy detection stage have been trained, we next focus on the privacy protection stage, which aims to minimize the image reconstruction distortion at the server as formulated in (\ref{O_s}). 
Specifically, once sensitive content is identified, the corresponding regions are removed from the captured images by the privacy segmentation module, producing sanitized images for transmission. 
The privacy protection stage then consists of two key components: a source-channel image transceiver that efficiently delivers the sanitized images to the server, and a VLM-based recovery module that reconstructs the removed content. 
The following subsections describe these two components in detail.


\subsubsection{Source-Channel Image Encoder and Decode}

The source-channel image encoder and decoder are implemented based on the DeepJSCC framework \cite{bourtsoulatze2019deep}, which adopts a fully convolutional autoencoder architecture for end-to-end wireless image transmission. Specifically, the encoder directly maps the sanitized image $\boldsymbol{A}_n$ into complex-valued channel symbols. The communication channel is modeled as a non-trainable differentiable layer within the network. 
At the receiver side, the decoder reconstructs the sanitized image $\boldsymbol{\hat{A}}_n$ from the received signal through a series of transpose convolutional layers. 
The encoder and decoder are jointly trained to minimize the reconstruction distortion of the sanitized image, and the loss function is defined as 
\begin{equation}
\mathcal{L}_S(\boldsymbol{\eta}, \psi)
=
\frac{1}{N}
\sum_{n=1}^{N}
\left\|
\boldsymbol{A}_n - \hat{\boldsymbol{A}}_n(\eta_n, \psi)
\right\|_2^2.
\label{LS}
\end{equation}
The training procedure is summarized in Algorithm~\ref{alg_deepjscc}.

\begin{algorithm}[t]
\caption{\label{alg_deepjscc}Train the Source-Channel Image Transceiver}
\begin{algorithmic}[1]\small

\State \textbf{Input}: 
$\mathcal{S}$, $\hat{i}$, 
$\tau_{\mathrm{thre}}^{(s)}$ and $P_{\max}^{(s)}$.

\While{$\mathcal{L}_S$ in (\ref{LS}) has not converged}

        \State Sample a multi-view image set
        $\{\boldsymbol{S}_n\}_{n=1}^{N}$ from $\mathcal{S}$.
        
        \For{$n=1,\ldots,N$}
        
            \State \textbf{Privacy Segmentation:}
            
            \State \hspace{0.5cm}
            $f_{\mathrm{SEG}}(\boldsymbol{S}_n| \hat{i})
            \rightarrow \boldsymbol{A}_n$.
            
            \State \textbf{Transmitter:}
            
            \State \hspace{0.5cm}
            $E_{\eta_n}
            (\boldsymbol{A}_n|
            \tau_{\mathrm{thre}}^{(s)},
            P_{\max}^{(s)})
            \rightarrow
            \boldsymbol{x}_n^{s}$.
            
            \State \hspace{0.5cm}
            Transmit $\boldsymbol{x}_n^{s}$ over the channel.
            
            \State \textbf{Receiver:}
            
            \State \hspace{0.5cm}
            Receive $\boldsymbol{y}_n^{s}$ via (\ref{ys}).
            
            \State \hspace{0.5cm}
            $D_{\psi}(\boldsymbol{y}_n^{s})
            \rightarrow
            \hat{\boldsymbol{A}}_n$.
        
        \EndFor
        
        \State Compute the loss $\mathcal{L}_S$ by (\ref{LS}).
        
        \State Update $\boldsymbol{\eta}$ and $\psi$ using gradient descent.

\EndWhile

\State \textbf{Output}: $\{E_{\eta_n}(\cdot)\}_{n=1}^{N}$ and $D_{\psi}(\cdot)$.

\end{algorithmic}
\end{algorithm}


\subsubsection{Privacy Recovery Module}

\begin{figure}[t] 
    \centering
    \includegraphics[width=0.48\textwidth]{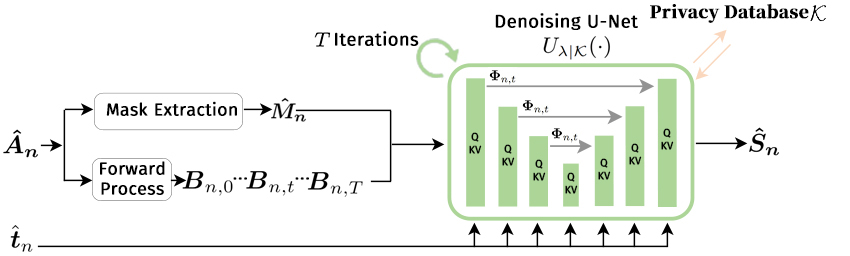}\vspace{-0.2cm}
    \caption{\label{fig_pr}Privacy Recovery Module} \vspace{-0.6cm}
\end{figure}

The privacy recovery module incorporates  a VLM-based diffusion framework to regenerate the removed sensitive content, as shown in Fig. 5. 
First, a binary mask \(\boldsymbol{\hat{M}_n} \in \{0,1\}^{H \times W} \) is obtained to detect the boundaries of the privacy region in $\boldsymbol{\hat{A}_n}$.

In parallel, the received image $\boldsymbol{\hat{A}_n}$ goes through a forward diffusion process, where Gaussian noise $\boldsymbol{\epsilon} \sim \mathcal{N}(\boldsymbol{0}, \mathbf{I})$ is gradually added over $T$ timesteps onto the image. 
At each step $t=0,\cdots, T$, the noisy latent variable is computed as:  
\begin{equation}
  \boldsymbol{B}_{n,t}   =   \sqrt{\bar{\alpha}}_t \cdot\boldsymbol{\hat{A}_{n}} + \sqrt{1 - \bar{\alpha}_t} \cdot \boldsymbol{\epsilon}  ,
\end{equation}
where  $ \bar{\alpha}_t \ge 0 $ is time-dependent hyperparameter that decreases with $t$. At $t=0$, $\bar{\alpha}_0$ equals to 1, so $\boldsymbol{{B}}_{n,0}=\boldsymbol{\hat{A}_n}$ contains no noise, and as $t$ increases, noise is progressively added, producing a sequence of latent variables $\{ \boldsymbol{B}_{n,t} \}_{t=0}^T$.

Next, a denoising U-Net $U_{\lambda|\mathcal{K}}(\cdot)$ is employed to reconstruct  the removed privacy content. 
The model parameters $\lambda$ are fine-tuned on the privacy database $\mathcal{K}$ via transfer learning, as summarized in Algorithm \ref{alg2}. 
The inputs to $U_{\lambda|\mathcal{K}}(\cdot)$ include the binary mask $\boldsymbol{\hat{M}_n}$ and the noisy latent sequence $\{ \boldsymbol{B}_{n,t} \}_{t=0}^T$  generated from the forward diffusion process. Meanwhile, the textual description $\hat{\boldsymbol{{t}}}_n$  provides conditional vision-language context to guide the image generation and ensure semantic consistency within the masked region.


The denoising process starts from timestep $t=T$, where the input $\boldsymbol{B}_{n,T}$ is processed by the U-Net to produce the intermediate output $\boldsymbol{C}_{n,T-1} = U_{\lambda|\mathcal{K}}(\boldsymbol{B}_{n,T}| \hat{\boldsymbol{{t}}}_n)$. 
The  generated content is then masked to retain only the privacy region, while the non-privacy background is filled using pixels from the corresponding noisy image $\boldsymbol{B}_{n,T-1}$, i.e.,
\begin{equation}
    \boldsymbol{D}_{n,T-1} =  \underbrace{\boldsymbol{C}_{n,T-1} \odot \boldsymbol{\hat{M}_n}}_{\text{\small Masked-region generation}} + \underbrace{\boldsymbol{B}_{n,T-1} \odot (\mathbf{1}- \boldsymbol{\hat{M}_n})}_{\text{\small Background retention}},
\end{equation}
which serves as the input for the next time step $t=T-1$. 
This process is repeated iteratively for  $t=T-1,\cdots,1$ via 
\begin{equation}
    \boldsymbol{D}_{n,t-1} = U_{\lambda|\mathcal{K}}(\boldsymbol{D}_{n,t}| \hat{\boldsymbol{{t}}}_n) \odot \boldsymbol{\hat{M}_n}  +  \boldsymbol{B}_{n,t-1} \odot (\mathbf{1}- \boldsymbol{\hat{M}_n}). 
\end{equation} 
The finial output $\boldsymbol{D}_{n,0}$ corresponds to the regenerated image $\boldsymbol{\hat{S}_n}$, with the removed privacy content semantically reconstructed by the receiver.

Furthermore, to enable conditional generation of the removed content based on the textual description $\hat{\boldsymbol{{t}}}_n$, the U-Net $U_{\lambda|\mathcal{K}}(\cdot)$ is augmented  with a cross-attention mechanism \cite{vaswani2017attention}.   
In each cross-attention layer, the query $\mathbf{Q}$ is computed from intermediate feature maps of the U-Net, while the key $\mathbf{K}$ and value $\mathbf{V}$ are derived from the text embedding of $\hat{\boldsymbol{{t}}}_n$: 
\begin{align}
\text{Attention}(\mathbf{Q}, \mathbf{K}, \mathbf{V}) &= 
\mathrm{softmax}\left( \frac{\mathbf{Q} \mathbf{K}^{T}}{\sqrt{\nu}} \right) \mathbf{V}, 
\end{align}
\begin{align}
\mathbf{Q} = \mathbf{W}_Q \cdot \boldsymbol{\Phi}_{n,t}, \quad 
\mathbf{K} = \mathbf{W}_K \cdot \hat{\boldsymbol{{t}}}_n, \quad 
\mathbf{V} = \mathbf{W}_V \cdot \hat{\boldsymbol{{t}}}_n, 
\end{align}
\noindent
where $\boldsymbol{\Phi}_{n,t}$ denotes the output of $n$-th intermediate layer within the U-Net at timestep $t$,  $\mathbf{W}_Q$, $\mathbf{W}_K$, $\mathbf{W}_V$ are learnable projection matrices, and $\nu$ is the scaling factor. The overall training procedure is provided in  Algorithm \ref{alg3}.

















\begin{algorithm}[t]
\caption{\label{alg2}Transfer Learning Based on Privacy Database $\mathcal{K}$}
\begin{algorithmic}[1]\small
\Statex \textbf{Initialization:} Load the pre-trained model $U_{\lambda}(\cdot)$
\State \textbf{Input:} $\mathcal{K} = \left\{ \mathcal{S}^{*}_i \right\}_{i=1,\cdots,K}$

\For{$i = 1, \cdots, K$}
    \State Sample a image set $\mathcal{S}^{*}_i = \{ \boldsymbol{S}^{*}_{i,m} \}_{m=1,\cdots,M}$ from $\mathcal{K}$
    \For{$m = 1, \cdots, M$}
        \State $f_{\mathrm{BLIP}}(\boldsymbol{S}^{*}_{i,m}) \rightarrow \boldsymbol{t}^{*}_{i,m}$
        \State Forward diffusion on $\boldsymbol{S}^{*}_{i,m}\rightarrow\{ \boldsymbol{B}_{t,i,m} \}_{t=0}^T$  
        \For{$t = T, \cdots, 1$}
            \State $U_{\lambda}(\boldsymbol{B}_{t,i,m}| \boldsymbol{t}^{*}_{i,m})\rightarrow\boldsymbol{C}_{t-1,i,m}$
            \State Compute  $\mathcal{L}_{t,i,m} = \| \boldsymbol{C}_{t-1,i,m} - \boldsymbol{B}_{t-1,i,m} \|^2$
            \State Update $\lambda$ using gradient descent on loss $\mathcal{L}_{t,i,m}$
        \EndFor
    \EndFor
\EndFor

\State \textbf{Output:} $U_{\lambda | \mathcal{K}}(\cdot)$

\end{algorithmic}
\end{algorithm}\normalsize

\begin{algorithm}[t]
\caption{\label{alg3}Train the Proposed Framework}
\begin{algorithmic}[1]\small

\State \textbf{Input}: 
$\mathcal{S}$, 
$\mathcal{K}$, 
$I_{\boldsymbol{\phi}}(\cdot)$,
$q_{\boldsymbol{\omega}}(\cdot)$,
$\tau_{\mathrm{thre}}^{(t)}$,
$\tau_{\mathrm{thre}}^{(s)}$,
$P_{\max}^{(t)}$,
$P_{\max}^{(s)}$,
and $\epsilon_k$.


\State Train the privacy detection module via Algorithms \ref{alg1} and \ref{alg_reid}.

\State Train the privacy recovery module via Algorithms \ref{alg_deepjscc} and \ref{alg2}.

\While{the overall objective in (\ref{equ_obj1}) has not converged}

    \State Sample a multi-view image set
    $\{\boldsymbol{S}_n\}_{n=1}^{N}$ from $\mathcal{S}$.

    \For{$n=1,\ldots,N$}

        \State Execute the privacy detection.
            \If{a privacy entity is detected}

                \State Execute the privacy protection stage.

            \Else

                \State Directly transmit the original images to the server.

            \EndIf

    \EndFor

   \State Compute the objective in (\ref{equ_obj1}).

    \State Update
    $\boldsymbol{\Theta}$, $\beta$, $\boldsymbol{\eta}$, $\psi$, and $\lambda$.

\EndWhile

\State \textbf{Output}: 
$\{T_{\boldsymbol{\theta}_n}(\cdot)\}_{n=1}^{N}$, 
$C_{\boldsymbol{\beta}}(\cdot)$, 
$\{E_{\boldsymbol{\eta}_n}(\cdot)\}_{n=1}^{N}$, $D_{\boldsymbol{\psi}}(\cdot)$,  $U_{\lambda|\mathcal{K}}(\cdot)$

\end{algorithmic}
\end{algorithm}

\section{Simulation Results and Analysis}\label{simulation}

To evaluate the  proposed privacy-preserving semantic framework, we use the multi-view dataset in \cite{wu20153dshapenets}, which contains  $3,926$ objects of $13$ categories, captured from $M$ different viewpoints. 
In the simulations, we set $M=3$, i.e., each object is observed from three viewpoints, with each viewpoint captured by a distinct edge device.  
To improve the realism of the dataset, diverse background environments and lighting conditions are generated using Blender \cite{blender}. 
For the text generation module, we adopt the pretrained Qwen2.5-VL model\cite{bai2025qwen25vl}, which employs a vision-language architecture to convert images into corresponding textual descriptions.  
For the ESCT transceiver, we adopt the pretrained T5-small model \cite{raffel2020t5} as the semantic text encoder and decoder, where the T5 encoder extracts semantic representations from the input sentence, and the T5 decoder reconstructs the sentence from the recovered representations. 
The ESCT encoder and decoder, as well as the source-channel image transceiver, are trained in an end-to-end manner over an AWGN channel with an  SNR of $12$ dB and communication threshold of $\epsilon_k=138$. 
Throughout the training and testing stages, the transmitted symbols are normalized as $P_{\max}^{(t)}=P_{\max}^{(s)}=1$ Watt. 
During testing, unless specified, the SNR thresholds are $\tau_{\mathrm{thre}}^{(t)}=\tau_{\mathrm{thre}}^{(s)}=0$ dB. 
For the Text-based Privacy Identification Module,  the pretrained all-MiniLM-L6-v2 model\cite{MiniLM} is adopt, where input text messages are projected into a $384$-dimensional identity embedding space.
The privacy database is constructed by randomly selecting 30 sensitive images from the training set. 
For the privacy segmentation module, binary masks of the target objects are generated during the blender rendering process and used for privacy removal. 
For the privacy recovery module, the stable diffusion model \cite{rombach2022high} is fine-tuned on the privacy dataset to reconstruct the removed private content, conditioned on text embeddings. 
The adversary is equipped with a similar decoder architecture and the same pre-trained Stable Diffusion model but has no access to the privacy database. After decoding the received signal, the adversary attempts to recover the removed private content using only the received data and its local generative model.

Four scenarios are considered in the performance evaluation of the proposed approach.   
Server-Full and Adversary-Full evaluate the reconstruction quality of the entire image at the authorized server and unauthorized receiver, respectively, 
while Server-Privacy and Adversary-Privacy evaluate the reconstruction quality in the sensitive region of the transmitted image at the authorized server and unauthorized receiver, respectively.  
To further compare the proposed approach with the state-of-the-art method, the DeepJSCC model \cite{bourtsoulatze2019deep} is adopted as the baseline, where the original image is transmitted directly without any privacy protection, and both authorized receiver and  adversary try to reconstruct the full image. For a fair comparison, the compression ratio, defined as the ratio of the dimension of the transmitted signal to that of the input image, is set to $1/48$ for both the proposed  approach and the DeepJSCC baseline. 

\begin{figure}[t]
    \begin{center}
        \begin{subfigure}{0.679\columnwidth}
            \centering
            \includegraphics[width=\linewidth]{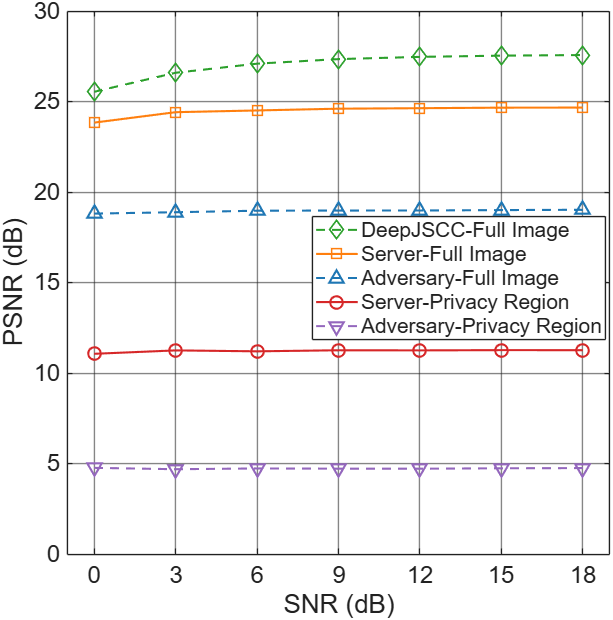}
            \caption{\label{fig_psnr}}
        \end{subfigure} 
        \begin{subfigure}{0.7\columnwidth}
            \hspace{-0.21cm}
            \includegraphics[width=\linewidth]{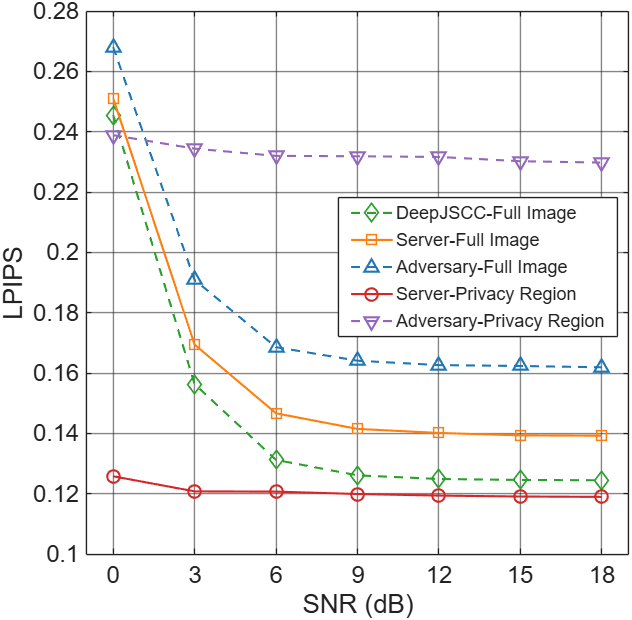}
            \caption{\label{fig_lpips}}
        \end{subfigure}
        \caption{\small
        (a) PSNR and (b) LPIPS of the original and reconstructed images at the server and the adversary.
        }\vspace{-0.65cm}
        \label{fig1}
    \end{center}
\end{figure}

As shown in Fig.~\ref{fig_psnr}, the peak signal-to-noise ratio (PSNR) is used to quantify the reconstruction quality, where a higher PSNR indicates greater pixel-level fidelity.  
The PSNR between the original image $\boldsymbol{S}$ and the reconstructed image $\hat{\boldsymbol{S}}$ in dB is 
\begin{equation}\small
\text{PSNR}(\boldsymbol{S},\hat{\boldsymbol{S}}) = 10 \log_{10} \left( \frac{X_{\max}^2}{\text{MSE}(\boldsymbol{S},\hat{\boldsymbol{S}})} \right)  ,
\end{equation}\normalsize
where $X_{\max}$ is the maximum pixel value and MSE denotes the mean squared error. 
Fig.~\ref{fig_psnr} shows that the PSNR of the reconstructed full image at the server increases with the received SNR, while that at the adversary remains nearly unchanged. This difference arises from their distinct capabilities in image reconstructions.
At the server, better channel condition improves the quality of received image  $\hat{\boldsymbol{A}}_n$, as well as the reliability of text transmission and physical-layer key generation, which jointly leads to a better 
reconstruction quality. 
In contrast, although the information received by the adversary improves with better channel conditions, it cannot access the privacy database, and thus its reconstruction of the privacy region shows little improvement. Consequently, the server consistently achieves a PSNR improvement of more than $6$ dB over the adversary. 
%
%
Compared with the DeepJSCC baseline, the proposed framework yields a lower PSNR by around $2$ dB, reflecting the tradeoff between reconstruction fidelity and privacy protection.  
Fig. \ref{fig_psnr} further shows that the PSNR within the privacy region remains nearly constant as the channel SNR increases for both the authorized server and the adversary. 
This is because the masked region contains no visual information, and thus improving the channel quality provides almost no benefits for image reconstruction in this region. 
Moreover, the PSNR within the privacy region is lower than that of the full image for both the server and the adversary, indicating the greater difficulty of pixel-level reconstruction in the masked area. 
Nevertheless, by leveraging the privacy database, the authorized server achieves a PSNR gain of more than $6$ dB over the adversary in the sensitive region, demonstrating the effectiveness of the proposed   framework. 

In Fig.~\ref{fig_lpips}, we further adopt the Learned Perceptual Image Patch Similarity (LPIPS) metric \cite{zhang2018unreasonable}  to evaluate the perceptual quality of the reconstructed images.  
Using a pretrained AlexNet model \cite{krizhevsky2012imagenet}, LPIPS measures perceptual similarity in the deep feature space rather than at the pixel level, by
\begin{equation}\small
\text{LPIPS}(\boldsymbol{S},\hat{\boldsymbol{S}})
=
\sum_{l} \frac{1}{H_lW_l} \sum_{h,w}
\left\|\boldsymbol{\delta}_{l}\odot
\left(f_{\mathrm{Alex}}^{\,l}(\boldsymbol{S})
-
f_{\mathrm{Alex}}^{\,l}(\hat{\boldsymbol{S}})
\right)\right\|_2^2 ,
\end{equation}\normalsize
where $f_{\mathrm{Alex}}^{\,l}(\cdot)$ denotes the normalized feature map extracted from the $l$-th layer of AlexNet, and $\boldsymbol{\delta}_l$ is the learned channel-wise weighting vector for that layer.
A lower LPIPS value corresponds to higher perceptual similarity, 
as it is computed from weighted feature distances.
Unlike PSNR, LPIPS aligns more closely with human visual perception. 
Fig. \ref{fig_lpips} shows that the LPIPS of
the reconstructed full images at both the server and adversary decrease with increasing channel SNR, as improved channel conditions enable more reliable recovery of semantic and visual information. 
As a result, the recovered image generated by the server in the proposed approach consistently achieves a lower LPIPS value by approximately $12\%$ than the adversary. 
Fig. \ref{fig_lpips} further shows that  LPIPS within the privacy region remains nearly unchanged as the  SNR increases at both the authorized server and the adversary, since the masked region carries almost no information. 
The reconstructed image generated by the authorized server  achieves the lowest LPIPS within the privacy region (Server-Privacy), indicating superior semantic consistency. 
Since the sensitive region is regenerated by the diffusion model rather than directly reconstructed from the received signal,  it is less susceptible to channel noise while better preserving semantic meaning. 
Moreover, by leveraging the privacy database, the authorized server (Server-Privacy) achieves a lower LPIPS by $48\%$ than the adversary (Adversary-Privacy) within the privacy region. 
Compared with the DeepJSCC baseline, although the proposed framework incurs a slight increase in the LPIPS of the full image, it significantly improves the perceptual quality on the reconstructed privacy region while effectively preventing privacy leakage, demonstrating a favorable tradeoff between semantic fidelity and privacy preservation.

\begin{figure*}[t]
	\begin{center} \vspace{-0.0cm}
		\begin{subfigure}{.32\textwidth}
			\centering
			\includegraphics[width=5.83cm]{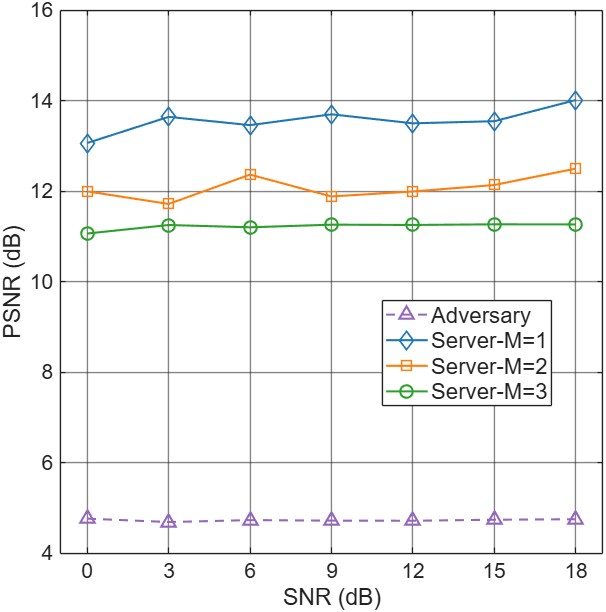}
			\caption{\label{fig_psnr_m} }
		\end{subfigure}
		\begin{subfigure}{.32\textwidth}
			\centering
			\includegraphics[width=5.99cm]{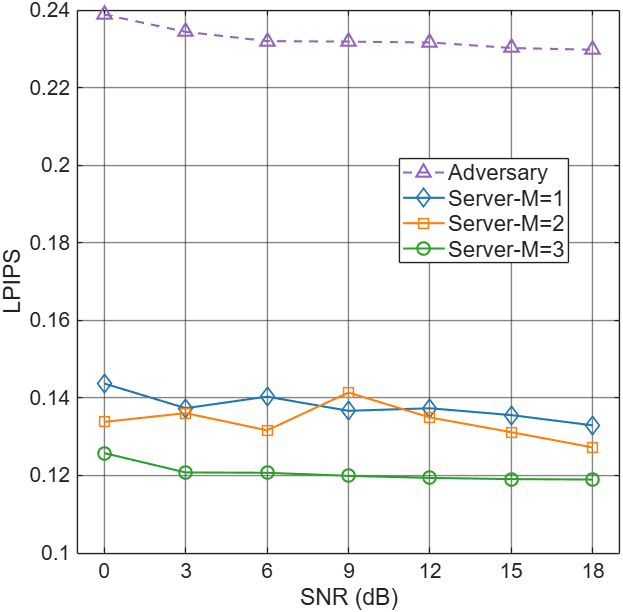}
			\caption{\label{fig_lpips_m} }
		\end{subfigure}
		\begin{subfigure}{.32\textwidth}
			\centering
			\includegraphics[width=5.83cm]{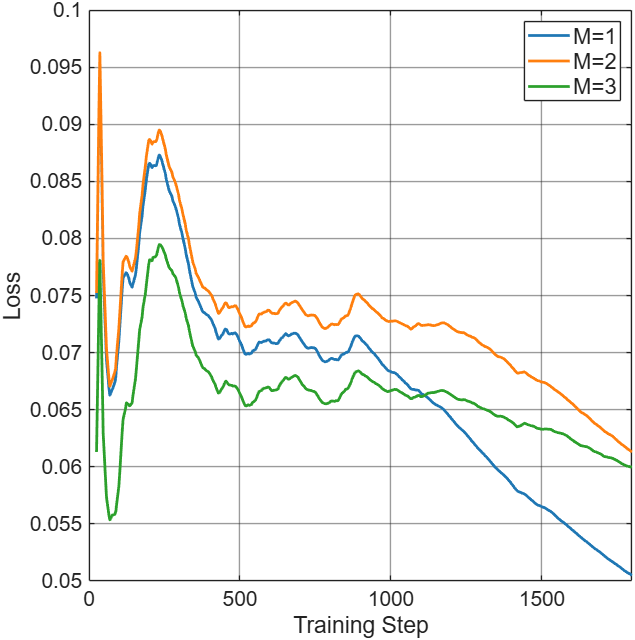}
			\caption{\label{fig_loss}  }
		\end{subfigure}
		\vspace{-0.0cm}
		\caption{\small{\label{fig2} 
        Effect of privacy database size on privacy-region reconstruction quality: (a) PSNR, (b) LPIPS, and (c) training loss.
            }
		}
	\end{center}
	\vspace{-0.2cm}
\end{figure*}

In Fig.~\ref{fig2}, we evaluate the impact of the privacy database size by varying the number of reference images, denoted by $M$, for each privacy entity in $\mathcal{K}$ during Stable Diffusion fine-tuning in Algorithm~\ref{alg2}. 
For example, $M=3$ corresponds to to three reference images captured from the front, side, and top views. 
For comparison, the adversary has no access to the privacy database $\mathcal{K}$, corresponding to $M=0$. 
As shown in Fig.~\ref{fig_psnr_m}, increasing $M$ from $1$ to $3$ slightly reduces the PSNR by approximately $1$ dB, since multiple viewpoints make pixel-wise reconstruction more challenging. 
In contrast, Fig.~\ref{fig_lpips_m} shows that the LPIPS decreases as  $M$ increases, indicating improved perceptual quality due to richer semantic information from multiple viewpoints. 
These results reveal a tradeoff between pixel-level fidelity and perceptual quality: incorporating additional viewpoints slightly degrades pixel-wise reconstruction accuracy but significantly improves perceptual consistency. 
More importantly, the server consistently outperforms the adversary across all values of $M$, demonstrating the effectiveness of the proposed privacy database $\mathcal{K}$.




Fig.~\ref{fig_loss} shows the training loss of Algorithm~\ref{alg2} during the privacy recovery optimization with a learning rate of $1\times10^{-6}$. 
Regardless of the viewpoint number $M$,  the proposed framework eventually converges after sufficient training. 
Fig.~\ref{fig_loss} first shows that a smaller $M$ leads to a lower final training loss, as a simpler training dataset is easier for the model to fit. 
However, a lower training loss does not guarantee a better perceptual reconstruction. 
The limited view diversity reduces the model's generalization capability, resulting in lower visual quality as shown in Fig.~\ref{fig_lpips_m}. 
In contrast, multi-view reference images provide richer semantic information during fine-tuning and a better perceptual reconstruction result. 
Therefore, increasing view diversity improves semantic fidelity, even though it slightly compromises pixel-level accuracy and training speed.

\begin{figure*}[t]
	\begin{center} \vspace{-0.0cm}
		\begin{subfigure}{.32\textwidth}
			\centering
			\includegraphics[width=5.7cm]{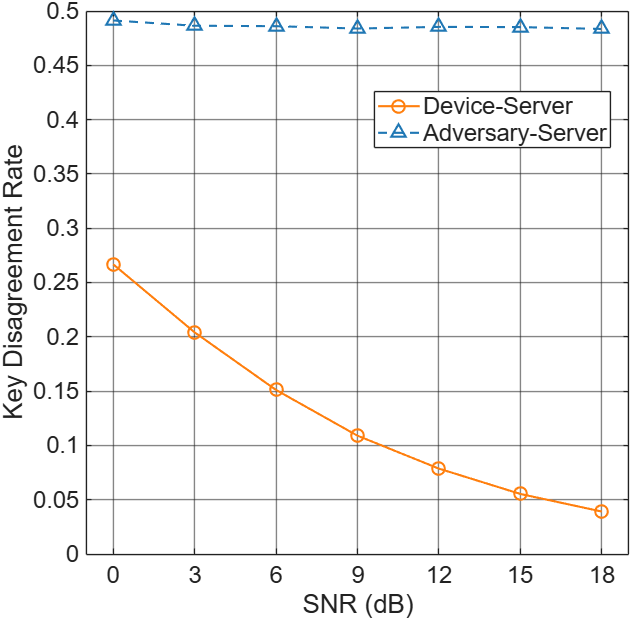}
			\caption{\label{fig_key} }
		\end{subfigure}
		\begin{subfigure}{.32\textwidth}
			\centering
			\includegraphics[width=5.58cm]{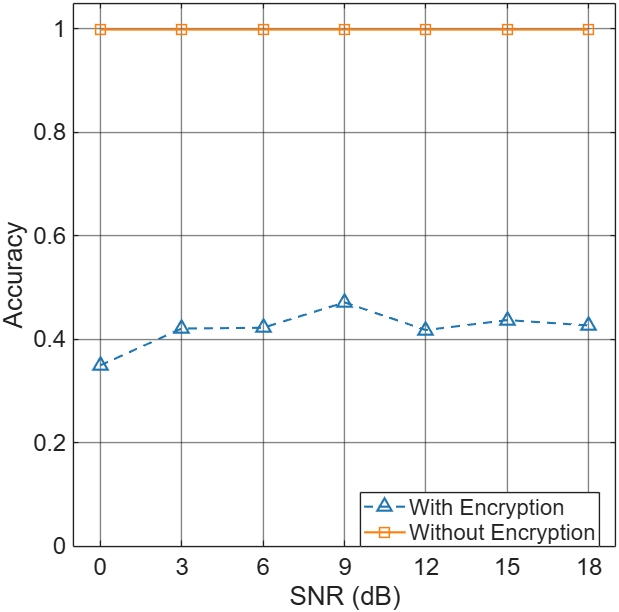}
			\caption{\label{fig_id} }
		\end{subfigure}
		\begin{subfigure}{.32\textwidth}
			\centering
			\includegraphics[width=6cm]{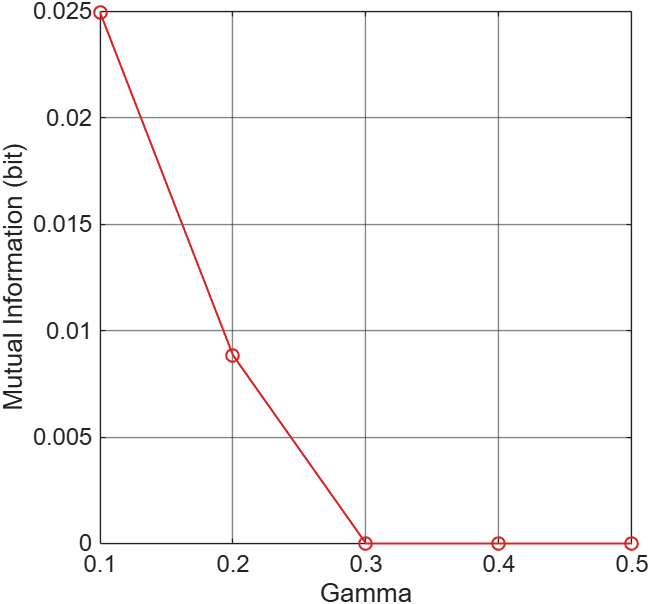}
			\caption{\label{fig_dsib}  }
		\end{subfigure}
		\vspace{-0.0cm}
		\caption{\small{\label{fig3} 
        (a) Key disagreement rate between the device and the server, and between the adversary and the server under different SNR. (b) Privacy detection accuracy at the adversary with and without encryption under different SNR. (c) Cross-view mutual information under different values of $\gamma$.   
            }
		}
	\end{center}
	\vspace{-0.2cm}
\end{figure*}

In Figs.~\ref{fig_key} and~\ref{fig_id}, we investigate the impact of the encryption text transmission on the privacy detection stage.   
Specifically, the key disagreement rate (KDR) is introduced to evaluate the bit-wise discrepancy between the key $\boldsymbol{k}^\text{dev}$ generated at the device and $\boldsymbol{k}^\text{srv}$ at the server,  by 
\begin{equation}\small
\mathrm{KDR}(\boldsymbol{k}^\text{dev},\boldsymbol{k}^\text{srv})
=
\frac{1}{d}\sum_{j=1}^{d}\mathbf{1}
    \left(
            {k}_j^\text{dev}
            \neq 
            {k}_j^\text{srv}
    \right),
\label{kdr}
\end{equation}\normalsize
where $k_{j}^{\text{dev}}$ and $k_{j}^{\text{srv}}$ denote the $j$-th element of corresponding keys,  
and the indicator function $\mathbf{1}(\cdot)$ outputs $1$ if  two elements equal and $0$ otherwise. A lower KDR indicates higher key consistency. 
In Fig. \ref{fig_key}, the KDR  between the device and the server decreases as the SNR increases. This is because a higher SNR results in better quality of the physical channel measurement, thereby improving the key generation accuracy.  
Meanwhile, we assume that the adversary intercepts the strongest pilot signals between the device and the server, so as to estimate the CSI and generate a secret key using the same function $f_{\mathrm{PLK}}(\cdot)$.
As shown in Fig. \ref{fig_key}, the adversary achieves a much higher KDR than the legitimate users, and the resulting KDR stays close to $0.5$ across all SNR, indicating that its generated keys are nearly random and statistically independent of the legitimate keys. 
This is because the adversary cannot exploit the channel reciprocity of the legitimate pair. 
These results demonstrate that the physical-layer key generation scheme achieves high key consistency between the legitimate uses while preventing adversary key recovery.

To further evaluate the privacy leakage during the text-transmission stage, we feed the intercepted text-only signal $\boldsymbol{y}_{a,n}$ to the adversary, and check if the privacy identification can be detected in the reconstructed message $D_{\beta_a}(\boldsymbol{y}_{a,n})$. 
As shown in Fig.~\ref{fig_id}, 
if no encryption is employed, the adversary achieves nearly $100\%$ privacy identification accuracy across all SNR levels, indicating that the transmitted text alone leaks sufficient semantic information to reveal the privacy identity. 
In contrast, enabling physical-layer encryption reduces the identification accuracy to below $50\%$. 
The results in Fig.~\ref{fig_id} further confirm the observations in Fig.~\ref{fig_key}, where a high KDR between the adversary and the legitimate nodes leads to incorrect text decryption, 
thereby effectively suppressing privacy leakage.  
These results demonstrate that the proposed physical-layer encryption effectively protects sensitive identity information during the text transmission stage, even against a strong model-aware adversary.


Fig.~\ref{fig_dsib} investigates the effect of the weighting factor $\gamma$ in (\ref{L_D_2}), which controls the trade-off between the cross-view mutual information minimization and other training purposes.  
A higher mutual information weight indicates that  more redundant semantic information will be removed from multiple devices' observation of the same event. 
Note that, in this calculation, the cross-view mutual information is estimated using the well-trained variational approximation $q_{\omega}$ with a neural-based approximation of the conditional distribution in (\ref{eq:vclub_bound}). 
Fig.~\ref{fig_dsib} shows that the cross-view mutual information decreases as $\gamma$ increases and approaches zero when $\gamma \ge 0.3$. 
This indicates that redundant semantic information across different viewpoints can be effectively removed by setting the weight $\gamma = 0.3$. 

\section{Conclusion}

In this paper, we have proposed a novel VLM-based framework for privacy-preserving semantic communications in wireless edge systems. 
A text-based privacy detection module has been developed to identify sensitive entities using a privacy database maintained exclusively at the server. 
To protect the textual information, we further developed an encrypted semantic-channel encoder/decoder with a practical a practical secret-key establishment mechanism. 
Simulation results show that  the keys generated from physical-layer characteristics exhibit a low mismatch rate between the device and server, while those obtained by the adversary remain effectively random. 
The proposed method also reduces textual privacy leakage by more than $50\%$ compared with unprotected transmission.  
For image transmission, sensitive regions identified by the privacy detection module are removed at the edge device and reconstructed at the server using the VLM and prior knowledge stored in the privacy database. 
Even with access to comparable model parameters, the adversary cannot accurately recover the sensitive content without access to the privacy database, resulting in a $48\%$ reconstruction-quality advantage at the server over the model-aware adversary. 
Furthermore, the proposed framework  effectively suppresses cross-device semantic redundancy,
with the estimated mutual information between transmitted representations approaching $0$ bit, indicating that redundant semantic information is rarely transmitted across devices.

\bibliographystyle{IEEEtran}
\bibliography{references}

\end{document}